\documentclass[12pt]{article}
\usepackage{amsmath,amssymb,amsthm,graphicx,amssymb,amsfonts,amsbsy,amsmath,amsthm,bigints,mathtools,float, pst-node,pst-text,hyperref,subfig,epstopdf, pst-3d,amsfonts,pstricks,color,soul,lineno}
\newtheorem{theorem}{Theorem}

\newtheorem{lemma}{Lemma}

\newtheorem{definition}{Definition}
\newtheorem{remark}{Remark}
\date{}
\input epsf.sty
\numberwithin{equation}{section}
\usepackage[skip=0pt]{caption}
\begin{document}
\begin{center}
\baselineskip .2in {\large\bf Dynamics of a non-autonomous system with age-structured growth and harvesting of prey and mutually interfering predator with reliance on alternative food} 
\end{center}
\begin{center}
\baselineskip .2in {\bf N. S. N. V. K. Vyshnavi Devi$^{1}$, Debaldev Jana$^{1}$ and M. Lakshmanan$^{2}$ }
{\small\it $^1$Department of Mathematics \& SRM Research Institute,}
\\ SRM Institute of Science and Technology, Kattankulathur-603 203, Tamil Nadu, India \\
{\small\it $^2$Centre for Nonlinear Dynamics,}
\\ School of Physics, Bharathidasan University, Tiruchirapalli - 620 024, India.\\
\end{center} 
\begin{abstract}

We perform a detailed analysis of the behaviour of a non-autonomous prey-predator model where age based growth with age discriminatory harvesting in prey and predator's reliance upon alternative food in the absence of that particular prey are considered. We begin by deriving certain sufficient conditions for permanence and positive invariance and then proceed to construct a Lyapunov function to derive some constraints for global attractivity. With the help of continuation theorem we arrive at the best fit criterion to prove the occurrence of a positive periodic solution. Moreover, using Arzela-Ascoli theorem we formulate a proof for a unique positive solution to be almost periodic and we carry out numerical simulation to verify the analytical findings. With the aid of graphs and tables we show the nature of the prey-predator system in response to the alternative food and delays.\\

\end{abstract}

\section{Introduction}\label{Sec:1}
A relationship that involves hunting of one species by another species, for the sake of feeding, is called a prey-predator relation. The organism that gains nutrition by consuming another organism is called a predator and the organism that is being fed upon is called the prey \cite{PPR18}. Prey and predator populations respond dynamically to one another. These relationships are vital in conserving the stability of the ecosystem.\\

\par Hilsa shad (\textit{Tenualosa ilisha} (coined by Hamilton, 1822)) is a popular marine fish found in the Indo-Pacific region. This fish is commonly called as hilsa. Hilsa feeds on various Phytoplanktons and Zooplanktons \cite{HKM16}. The size at sexual maturity attained by male and female Hilsa vary from 160-400 mm and 190-430 mm respectively \cite{PVN58}. When fish reproduces for the first time (mostly when it turns one) it is concluded that it has attained maturity \cite{BU15K, SAV12}. Hilsa is said to be anadromous in nature; adults migrate from marine water to freshwater for spawning; the young rear in the freshwater before moving into the marine water for further feeding and growth. Usually found in the Bay of Bengal \cite{BU15}, for the sake of breeding and feeding, Hilsa migrates to the Ganges-Brahmaputra-Meghna (GBM) river basin, and the trans-boundary river basin is divided between India, China, Nepal, Bangladesh and Bhutan. This anadromous migration takes place during the monsoon season (July to October).\\

\par The striped spiny eel (\textit{Macrognathus pancalus} (Hamilton, 1822)), also called barred spiny eel, is a freshwater fish that is mainly noticed in rivers Ganga and Brahmaputra (in India and Bangladesh) and other countries \cite{VW10, SJ94}. They are catadromous; adult fishes migrate into marine waters from fresh water for spawning; younger fishes grow in the sea and later move to the rivers for food and development. Eel fishes are carnivores or meat eaters. Eel fishes feed on small forage fishes, aquatic insects, annelids, etc. The intensity of feeding is very high in early maturity and post-spawning periods. Younger fish consume more food in monsoon season and continue to do so post-monsoon season \cite{SMR05}. \\

\par Note that the younger specimens of eel are usually in the marine waters and the young matured hilsa and adult hilsa are also in the marine waters. This leads to the younger specimens of eel feeding on the young matured hilsa, and thus hilsa forms an important part of eel's diet. As mentioned earlier, consumption of food in younger eels is higher in monsoon but, hilsa migrates to GBM river basin during the monsoon season for spawning. Thus, during the monsoon, eel fishes are forced to depend upon alternative food, although during other seasons eel gets to feed on hilsa. We are particularly interested in modeling this prey-predator relationship between hilsa and eel and also the dependence of eel upon alternative food during the monsoon season when hilsa migrates to the GBM river basin.\\

\par Over exploitation of hilsa at the GBM basin is an alarming issue. The young matured hilsa that migrate to the GBM basin for spawning might be harvested before they spawn. Also, large stocks of small and immature hilsa are being harvested from GBM basin which is a major cause of concern \cite{SAV12}. It is a fact that hilsa cannot attain maturity instantaneously. The time required to attain maturity can be called as maturity delay. If we continue harvesting hilsa in its pre-maturation stage, then hilsa cannot be saved from becoming extinct. One preventive measure is to implement the age structured model of hilsa population.\\
\par The standard population equilibrium model, the Verhulst’s growth model \cite{V38}, explains population growth in a restricted environment. It is reported that the rate at which the population changes is dependent on the following aspects: growth and death components and intraspecific competition, where the growth factor is embedded with delay $\varsigma$ in it \cite{AWW06}. Also it was assumed that the death rate and intraspecific competition rate both together contribute to the rate of decline, which is instantaneous. Rates of death and intraspecific competition are specified by linear and quadratic terms respectively. Now we assume that individuals of hilsa fish which are born before $\varsigma$ time from present instant \textit{t} are sexually able to give birth to new offspring. Jana, Dutta and Samanta \cite{JDS19} considered $\chi(t) = \frac{b\varphi \chi(t-\varsigma)}{\varphi e^{\varphi\varsigma}+ c(e^{\varphi\varsigma}-1)\chi(t-\varsigma)}$ to be the density of hilsa fish that can give birth at time \textit{t} \cite{AWW06}. $\chi(t)$ denotes the number of individuals of prey (hilsa) at time $t$. $b$ represents birth rate of hilsa, $\varphi$ is the natural mortality rate in hilsa. $c$ indicates removal due to intraspecific competition among hilsa.\\

\par Eel being the predator, it consumes hilsa at the rate $\frac{\nu \zeta \chi(t)\Upsilon(t)}{a+\zeta \chi(t)+(1-\zeta)A+\xi \Upsilon (t)}$ with Beddington-DeAngelis type functional response \cite{BJR75,CRS01,JD14,LSB10}, where $\Upsilon(t)$ gives the number of predator (eel) at time $t$. The probability that eel fish depends upon hilsa for feeding is given by $\zeta $ and $(1-\zeta)$ is the probability that eel fish depends upon alternative food $A$ \cite{BM01,SS16} when hilsa migrates to freshwater for spawning. This explains the significance of the term $\zeta \chi (t)+(1-\zeta)A$ in the denominator of the above functional response. If $\zeta=1$ (eel feeds only on hilsa) i.e, the probability that eel depends on alternative food $(1-\zeta)$ is zero, then the functional response term becomes $\frac{\nu \chi (t) \Upsilon(t)}{a+\chi(t)+\xi \Upsilon(t)}$. If $\zeta=0$ (eel does not feed on hilsa) i.e, eel feeds only on alternative food $A$, then $\frac{\nu \zeta \chi (t) \Upsilon (t)}{a+\zeta \chi (t)+(1-\zeta)A+\xi \Upsilon (t)}$ term vanishes \cite{GC18}. $\xi \Upsilon (t)$ measures the mutual interference among eel, $a$ is the saturation constant and $\nu$ denotes the feeding rate of eel. If $\kappa$ is the coefficient of growth in eel population such that $\kappa<\nu$ then the eel population grows at the rate $\frac{\kappa \{\zeta \chi(t)+(1-\zeta)A\}\Upsilon (t)}{a+\zeta \chi(t)+(1-\zeta)A+\xi \Upsilon (t)}$ \cite{CRS01}. If $m$ is the natural mortality in eel and $m\Upsilon (t)$ gives the decline rate of eel, hilsa-eel model becomes

\begin{equation}\label{eq:1}
\begin{split}
\chi'(t)&=\dfrac{b\varphi \chi(t-\varsigma)}{\varphi e^{\varphi\varsigma}+c(e^{\varphi\varsigma}-1)\chi(t-\varsigma)}-
\varphi \chi(t)-c\chi^2(t)-\dfrac{\nu \zeta \chi(t) \Upsilon (t)}{a+\zeta \chi(t)+(1-\zeta)A+\xi \Upsilon (t)},\\
\Upsilon '(t)&=\dfrac{\kappa \{\zeta \chi(t)+(1-\zeta)A\}\Upsilon (t)}{a+\zeta \chi(t)+(1-\zeta)A+\xi \Upsilon (t)}-m\Upsilon (t),
\end{split}
\end{equation}
where $a$, $b$, $c$, $m$, $\zeta $, $A$, $\nu$, $\kappa$, $\varphi$ and $\xi$ are all positive constants. \\ 

\par Several authors \cite{SPH11, SBE16} reported that the temperature of water plays a crucial role as a spawning stimulant. This shows the dependence of birth rate on water temperature. Also, water level, turbidity, heavy rains or no frequent rains, air temperature, etc are some factors that affect the birth, growth, harvest and death in fishes \cite{SJ59}. Some physical, chemical/biochemical factors like water temperature, pH, alkalinity, hardness, etc are important to fish growth and mortality \cite{VRC05}. All the above mentioned factors are time variant. Hence, the system parameters must be time variant.
Incorporating the time factor into the parameters of the model \eqref{eq:1}, we get a non-autonomous system as follows,
\begin{equation}\label{eq:2}
\begin{split}
\chi'(t)&=\dfrac{b(t)\varphi(t) \chi(t-\varsigma)}{\varphi(t) e^{\varphi(t)\varsigma}+c(t)(e^{\varphi(t)\varsigma}-1)\chi(t-\varsigma)}-\varphi(t) \chi(t)-c(t)\chi^2(t)\\
&-\dfrac{\nu(t)\zeta(t)\chi(t) \Upsilon (t)}{a(t)+\zeta(t)\chi(t)+(1-\zeta(t))A(t)+\xi(t) \Upsilon (t)},\\
\Upsilon '(t)&=\dfrac{\kappa(t) \{\zeta(t)\chi(t)+(1-\zeta(t))A(t)\} \Upsilon (t)}{a(t)+\zeta(t)\chi(t)+(1-\zeta(t))A(t)+\xi(t) \Upsilon (t)}-m(t) \Upsilon (t).
\end{split}
\end{equation}

\par Hilsa forms an essential component of the inland fish resource of India \cite{MBP11}. Socially and culturally Hilsa is very important to people living in West Bengal, Odisha, etc., states in India \cite{S17}. The abundance of hilsa is declining each year owing to the excessive fishing of juveniles. In view of its economic importance, there is a need for enhancing the stock of hilsa by adopting sustainable management practices \cite{MBP11}. It is beneficial if hilsa is harvested when it attains certain body size and weight after spawning. This time delay can be called as delay in harvest. Let us say that we can harvest
$\frac{q(t)E(t)\varphi(t)\chi(t-\varsigma)}{\varphi(t) e^{\varphi(t)\varsigma}+c(t)(e^{\varphi(t)\varsigma}-1)\chi(t-\varsigma)}$ \cite{JDS19} density of hilsa fish at time $t$, where both the catchability coefficient ($q$) and fishing effort ($E$) are time variant. Properties of fishing equipment, extent of fishing, etc details are taken care of while considering the value of $q$ \cite{ASF96}. Fishing effort $E$ represents the number of fishermen, number of active boats, etc \cite{RAD01}. If $\varsigma_1$ and $\varsigma_2$ are considered to be delays in maturity and harvest respectively
i.e, individuals of hilsa born before $\varsigma_1$ time from present instant $t$ are able to give birth and we harvest for our economic profit those hilsa that are born before $\varsigma_2$ time from now ($t$), therefore system \eqref{eq:2} is written as follows 
\begin{equation}\label{eq:k}
\begin{split}
\chi'(t)&=\dfrac{b(t)\varphi(t) \chi(t-\varsigma_1)}{\varphi(t) e^{\varphi(t)\varsigma_1}+c(t)(e^{\varphi(t)\varsigma}-1)\chi(t-\varsigma_1)}-\varphi(t) \chi(t)-c(t)\chi^2(t)\\
&-\dfrac{\nu(t)\zeta(t)\chi(t) \Upsilon (t)}{a(t)+\zeta(t)\chi(t)+(1-\zeta(t))A(t)+\xi(t) \Upsilon (t)}-\dfrac{q(t)E(t)\varphi(t)\chi(t-\varsigma_2)}{\varphi(t) e^{\varphi(t)\varsigma_2}+c(t)(e^{\varphi(t)\varsigma_2}-1)\chi(t-\varsigma_2)},\\
\Upsilon '(t)&=\dfrac{\kappa(t) \{\zeta(t)\chi(t)+(1-\zeta(t))A(t)\} \Upsilon (t)}{a(t)+\zeta(t)\chi(t)+(1-\zeta(t))A(t)+\xi(t) \Upsilon (t)}-m(t) \Upsilon (t).
\end{split}
\end{equation}

During the last two decades several authors \cite{CS08, CT06, FK04, HY15, ZF08} have discussed non-autonomous dynamical systems that follow predation of Beddington-DeAngelis type. The primary objective of the present research is to study various dynamical aspects like permanence, positive invariance and global attractivity of the non-autonomous system \eqref{eq:2} which follows predation of Beddington-DeAngelis type and age based growth of the prey and to show the occurence of positive solutions that are periodic and almost periodic in nature. This is followed by the study of dynamical aspects of system \eqref{eq:k} that incorporates two delays, namely, delay in maturity ($\varsigma_1$) and delay in harvest ($\varsigma_2$) in prey.\\

\par We systematize this study into ten sections. Basic preliminaries are provided in Section 2. For system \eqref{eq:2}, positive invariance, permanence and global attractiveness are shown in sections 3, 4 and 5, respectively. We make it evident that there occurs positive solutions to system \eqref{eq:2} that are almost periodic and periodic, in sections 6 and 7. Section 8 deals with the dynamical aspects of system \eqref{eq:k}. Graphical interpretation of numerical examples is illustrated in section 9. Section 10 marks the conclusion of this paper.

\section{Preliminaries}\label{sec:2}
This section presents some basic definitions and assumptions.
\begin{enumerate}
\item[i.] We consider only those solutions $(\chi(t), \Upsilon (t))$ $\ni$ $\chi(t_0)$ and $\Upsilon(t_0)$ are positive $\forall$ $t_0$(initial value of time) $> 0$ for biological reasons.
\item[ii.] We consider $a(t)$, $b(t)$, $c(t)$, $m(t)$, $\zeta(t)$, $A(t)$, $\nu(t)$, $\kappa(t)$, $\varphi(t)$ and $\xi(t)$ to be bounded (positive upper and lower bounds) and continuous.
\item[iii.] If $r(t)$ is a bounded (by positive constants) and continuous function on $\mathbb{R}$,
\begin{center} \vspace{-0.2cm}
$r_L=\inf\limits_{t\in\mathbb{R}}$ $r(t)$ and $r_M=\sup\limits_{t\in\mathbb{R}}$ $r(t)$. \vspace{-0.2cm}
\end{center}
\item[iv.] The co-efficients of system \eqref{eq:2} must satisfy
\begin{center} \vspace{-0.2cm}
min$\{ b_L,a_L,c_L,\zeta_L,A_L,m_L,\nu_L,\kappa_L,\varphi_L,\xi_L\}>0$ and 
max$\{ b_M,a_M,c_M,\zeta_M,A_M,m_M,\nu_M,\kappa_M,\varphi_M,\xi_M\}<\infty$.
\end{center} 
\end{enumerate}
\begin{definition} 
The set $\{(\chi(t), \Upsilon (t))/$ $||(\chi(t), \Upsilon (t))-(\chi^\ast, \Upsilon ^\ast)|| \leqslant \varepsilon \}$ with center at $(\chi^\ast, \Upsilon ^\ast)$ $\ni$ every trajectory $(\chi(t), \Upsilon (t))$ of system \eqref{eq:2}, which is lying in the set at $t=t_0$(initial time), remains in it $\forall$ $t\geqslant t_0$(initial time) and tends to $(\chi^\ast, \Upsilon ^\ast)$ as $t\to\infty$. Such a set is said to be positively invariant with reference to system \eqref{eq:2}.
\end{definition}
\begin{definition} 
System \eqref{eq:2} is permanent when $\varphi>0$ and $\varphi>0$ (constants) exist and $0<\varphi\leqslant\varphi$ $\ni$
\begin{center} \vspace{-0.2cm}
$\min\{\lim\limits_{t\to +\infty}$ inf $\chi(t),\lim\limits_{t\to +\infty}$ inf $\Upsilon (t)\}\geqslant\varphi$ and $\max\{\lim\limits_{t\to +\infty}$ sup $\chi(t),\lim\limits_{t\to +\infty}$ sup $\Upsilon (t)\}\leqslant\varphi$ \vspace{-0.2cm}
\end{center}
for all those solutions of the system that originate from positive values (at $t=t_0$).
\end{definition}
\begin{definition} 
$(\chi^\ast(t), \Upsilon ^\ast(t))$, a bounded positive solution to system \eqref{eq:2} is globally attractive when $(\chi(t), \Upsilon (t))$ (another solution of the system) satisfies \begin{center} \vspace{-0.2cm}
$\lim\limits_{t\to\infty}$ $(|\chi(t)-\chi^\ast(t)|+| \Upsilon (t)- \Upsilon ^\ast(t)|)=0$. \vspace{-0.2cm}
\end{center}
\end{definition} 
\begin{definition} 
The set of all solutions of system \eqref{eq:2} is ultimately bounded if there is a $D>0$ $\ni$ for any solution $(\chi(t), \Upsilon (t))$ of \eqref{eq:2}, $\exists$ $W>0$ $\ni$ $||(\chi(t), \Upsilon (t))||\leqslant D$ for every $t\geqslant$ (initial time)$t_0$ $+ W.$
\end{definition}
\begin{definition} 
Consider the collection of continuous functions $C(D)$ on a compact metric space $D$ and let $j \in B\subseteq C(D)$. The set of functions B is called an equi-continuous family if for every $ \varepsilon>0$ $\exists$ $\eta>0\ni$ $d(x,y)<\eta \iff $ $|f(x)-f(y)|<\varepsilon$ for every $j$ that belongs to $B$.
\end{definition} 
\begin{definition} 
If a function $f:\mathbb{R}\rightarrow\mathbb{R}$ is such that $f(t+T)=f(t)$ is accurately true for immeasurably great number of values of $T$, where $T$ takes any value from $-\infty$ to $+\infty$ in such a way that arbitrarily long empty intervals are not left, then $f$ is said to be almost periodic.
\end{definition} 
\begin{definition} 
A function $z_1:\mathbb{R}\rightarrow\mathbb{R}$ is asymptotically (as $t\rightarrow\infty$) almost periodic if $z_1(t)=z_2(t)+z_3(t)$ where $z_2(t)$ and $z_3(t)$ are almost periodic and continuous functions respectively, $\ni$ $z_3$ is a real number and function $ z_3(t)\rightarrow 0$ as $ t$(time) $\rightarrow\infty.$
\end{definition}

\section{Positive Invariance}\label{Sec:3}
\begin{theorem}\label{thm:t1}
If $ b_L\varphi_L\xi_L > \nu_M \zeta_M $, $\kappa_M>m_L$ and $(\kappa_L-m_M)\zeta_L k^{\varepsilon}_{\chi}> m_M(a_M+(1-\zeta_L)A_M)$ then the set defined by 
\begin{center}
$ \Gamma_\varepsilon=\{(\chi(t), \Upsilon (t))\in \mathbb{R}^2 / k^{\varepsilon}_{\chi} \leqslant \chi(t) \leqslant K^{\varepsilon}_{\chi}, k^{\varepsilon}_{\Upsilon } \leqslant \Upsilon (t) \leqslant K^{\varepsilon}_{\Upsilon } \} $
\end{center}
is positively invariant 
w.r.t system \eqref{eq:2} where
\begin{center}
$ K^{\varepsilon}_{\chi} =\dfrac{b_M\varphi_M}{c_L}+\varepsilon$ ,
$ k^{\varepsilon}_{\chi}=\dfrac{b_L\varphi_L\xi_L-\nu_M \zeta_M }{c_M \xi_L} -\varepsilon $,
$K^{\varepsilon}_{\Upsilon } =\dfrac{(\kappa_M-m_L)\zeta_M K^{\varepsilon}_{\chi}} {m_L \xi_L}+\varepsilon$, \vspace{0.4cm}\\ $k^{\varepsilon}_{\Upsilon } =\dfrac{(\kappa_L-m_M)\zeta_L k^{\varepsilon}_{\chi}-m_M(a_M+(1-\zeta_L)A_M}{m_M \xi_M}-\varepsilon$,
\end{center}
and $ \varepsilon \geqslant 0 $ is sufficiently small so that $ k^{\varepsilon}_{\chi} > 0 . $
\end{theorem}
\begin{proof}
Suppose, a solution of system \eqref{eq:2} is given by $ (\chi(t), \Upsilon (t)) $ (with positive initial conditions $ \chi(t_0)$ and $\Upsilon (t_0)$)
with $ k^{\varepsilon}_{\chi} \leqslant \chi(t_0) \leqslant K^{\varepsilon}_{\chi} $ and $ k^{\varepsilon}_{\Upsilon } \leqslant \Upsilon (t_0) \leqslant K^{\varepsilon}_{\Upsilon } $. That is, $ (\chi(t_0), \Upsilon (t_0)) \in \Gamma_\varepsilon $. \\
From system \eqref{eq:2}, when $ t \geqslant t_0 $ we obtain that 
\begin{center}
$ \chi'(t) \leqslant b_M\varphi_M\chi(t-\varsigma)- c_L \chi^2 (t) $.
\end{center}
Rearranging the terms, we get, 
\begin{center}
$\dfrac{\chi'(t)}{c_L \chi(t)} \leqslant 
\dfrac{b_M\varphi_M}{c_L}\dfrac{\chi(t-\varsigma)}{\chi(t)}-\chi(t) $. 
\end{center}
Then, for every $ t \geqslant t_0$(initial time), 
\begin{center}
$\chi(t)\leqslant\dfrac{b_M\varphi_M}{c_L}\cdot$
\end{center}
It means that for any $\varepsilon >0 $ we can say 
\begin{center}
$\chi(t)\leqslant\dfrac{b_M\varphi_M}{c_L}+\varepsilon=K^{\varepsilon}_{\chi}.$
\end{center}
From system \eqref{eq:2}, when $ t\geqslant t_0 $ we obtain
\begin{center}
$ \chi'(t)\geqslant b_L\varphi_L\chi(t-\varsigma) - c_M \chi^2(t)-\dfrac{\nu_M \zeta_M \chi(t)}{\xi_L}\cdot $
\end{center}
This is same as
\begin{center}
$ \dfrac{\chi'(t)}{\chi(t)} \geqslant b_L\varphi_L \dfrac{\chi(t-\varsigma)}{\chi(t)}- c_M \chi(t) -\dfrac{\nu_M \zeta_M}{\xi_L}\cdot $
\end{center}
We can say that
\begin{center}
$ \chi(t) \geqslant \dfrac{b_L\varphi_L }{c_M}-\dfrac{\nu_M \zeta_M}{\xi_L c_M}\cdot $
\end{center}
And hence we get,
\begin{center}
$ \chi(t) \geqslant \dfrac{b_L\varphi_L\xi_L -\nu_M \zeta_M}{\xi_L c_M}\cdot $ 
\end{center}
For any $ \varepsilon>0 $ we can say that
\begin{center}
$ \chi(t) \geqslant \dfrac{b_L\varphi_L\xi_L -\nu_M \zeta_M}{\xi_L c_M}-\varepsilon=k^{\varepsilon}_{\chi}. $ 
\end{center}
Hence $ k^{\varepsilon}_{\chi} \leqslant \chi(t) \leqslant K^{\varepsilon}_{\chi} $ for every $ t \geqslant t_0$. \\
From system \eqref{eq:2}, when $ t \geqslant t_0 $ we obtain
\begin{center}
$ \Upsilon '(t) \leqslant \dfrac{\kappa_M \zeta_M K^{\varepsilon}_{\chi} \Upsilon (t)}{\xi_L \Upsilon (t)+K^{\varepsilon}_{\chi}\zeta_M}-m_L \Upsilon (t) $.
\end{center}
This means
\begin{center}
$ \Upsilon '(t) \leqslant \dfrac{(\kappa_M \zeta_M-m_L \zeta_M) K^{\varepsilon}_{\chi}\Upsilon (t)-m_L \xi_L \Upsilon ^2(t)}{\xi_L \Upsilon (t)+K^{\varepsilon}_{\chi}\zeta_M}\cdot $
\end{center}
For any $\varepsilon >0 $ we can say,
\begin{center}
$ \Upsilon '(t) \leqslant [K^{\varepsilon}_{\Upsilon }-\Upsilon (t)]\dfrac{m_L\xi_L \Upsilon (t)}{\xi_L \Upsilon (t)+K^{\varepsilon}_{\chi}\zeta_M}\cdot $
\end{center}
Hence for every $t\geqslant t_0$, 
\begin{center}
$ \Upsilon (t) \leqslant K^{\varepsilon}_{\Upsilon }. $
\end{center}
From system \eqref{eq:2}, when $ t \geqslant t_0 $ we obtain 
\begin{center}
$ \Upsilon '(t) \geqslant \dfrac{\kappa_L \zeta_L k^{\varepsilon}_{\chi} \Upsilon (t)}{a_M+\zeta_L k^{\varepsilon}_{\chi}+(1-\zeta_L)A_M+\xi_M \Upsilon (t)}-m_M \Upsilon (t) $.
\end{center}
Hence we get,
\begin{center}
$ \Upsilon '(t) \geqslant \dfrac{\Bigg[\dfrac{[(\kappa_L \zeta_L-m_M \zeta_L)k^{\varepsilon}_{\chi}-m_M(a_M+A_M(1-\zeta_L))]}{m_M\xi_M}-\Upsilon (t)\Bigg] \Upsilon (t)m_M \xi_M}{a_M+\zeta_L k^{\varepsilon}_{\chi}+(1-\zeta_L)A_M+\xi_M \Upsilon (t)}\cdot $ 
\end{center}
This implies that
\begin{center}
$\Upsilon '(t) \geqslant[k^{\varepsilon}_{\Upsilon }-\Upsilon (t)]\dfrac{\Upsilon (t)m_M \xi_M}{a_M+\zeta_L k^{\varepsilon}_{\chi}+(1-\zeta_L)A_M+\xi_M \Upsilon (t)}\cdot$ 
\end{center}
Hence for every $t\geqslant t_0$, 
\begin{center}
$ k^{\varepsilon}_{\Upsilon } \leqslant \Upsilon (t). $
\end{center}
Hence $ k^{\varepsilon}_{\Upsilon } \leqslant \Upsilon (t) \leqslant K^{\varepsilon}_{\Upsilon } $ for every $ t$ $ \geqslant $ (initial time)$t_0$.
That is, $ (\chi(t), \Upsilon (t)) \in \Gamma_\varepsilon $ for every $ t $ $\geqslant$ (initial time)$t_0$. \\
Hence $ \Gamma_\varepsilon $ is positively invariant w.r.t system \eqref{eq:2}.
\end{proof}

\section{Permanence}\label{Sec:4}
\begin{theorem}\label{thm:t2}
If $ b_L\varphi_L\xi_L > \nu_M \zeta_M $, $\kappa_M>m_L$ and $(\kappa_L-m_M)\zeta_L k^{0}_{\chi}> m_M(a_M+(1-\zeta_L)A_M)$\vspace{0.2cm}\\ hold then system \eqref{eq:2} is permanent, where $k^{0}_{\chi}=\dfrac{b_L\varphi_L\xi_L-\nu_M \zeta_M }{c_M \xi_L}\cdot$
\end{theorem}
\begin{proof}
As seen in Theorem \ref{thm:t1}, from system \eqref{eq:2} we have 
\begin{center}
$ \chi'(t)< b_M\varphi_M\chi(t-\varsigma)-c_L \chi^2(t)$.
\end{center}
Hence 
\begin{center}
$ \lim\limits_{t\rightarrow+\infty} $ sup $ \chi(t) \leqslant \dfrac{b_M\varphi_M}{c_L}=K^{0}_{\chi}. $
\end{center}
From system \eqref{eq:2} we also have
\begin{center}
$ \chi'(t)>b_L\varphi_L\chi(t-\varsigma) - c_M \chi^2(t)-\dfrac{\nu_M \zeta_M \chi(t)}{\xi_L}\cdot $
\end{center}
Since $b_L\varphi_L\xi_L>\nu_M \zeta_M$, we can say that 
\begin{center}
$ \lim\limits_{t\rightarrow+\infty} $ inf $ \chi(t) \geqslant \dfrac{b_L\varphi_L\xi_L-\nu_M \zeta_M}{c_M \xi_L}=k^{0}_{\chi}. $
\end{center}
From system \eqref{eq:2} we have 
\begin{center}
$ \Upsilon '(t)<\dfrac{\kappa_M \zeta_M K^{0}_{\chi} \Upsilon (t)}{\xi_L \Upsilon (t)+K^{0}_{\chi}\zeta_M}-m_L \Upsilon (t)$.
\end{center}
Hence
\begin{center}
$ \lim\limits_{t\rightarrow+\infty} $ sup $ \Upsilon (t) \leqslant \dfrac{(b_M-m_L)\zeta_M K^{0}_{\chi}}{m_L\xi_L}=K^{0}_{\Upsilon}. $
\end{center}
From system \eqref{eq:2} we have
\begin{center}
$ \Upsilon '(t)> \dfrac{\kappa_L \zeta_L k^{0}_{\chi} \Upsilon (t)}{a_M+\zeta_L k^{0}_{\chi}+(1-\zeta_L)A_M+\xi_M \Upsilon (t)}-m_M \Upsilon (t) $.
\end{center}
Thus we have
\begin{center}
$ \lim\limits_{t\rightarrow+\infty} $ inf $ \Upsilon (t) \geqslant \dfrac{(\kappa_M-m_L) \zeta_M k^{0}_{\chi}-m_M(a_M+A_M(1-\zeta_L)}{m_M\xi_M}=k^{0}_{\Upsilon}.$
\end{center}
Since, all the conditions in the definition of permanence are satisfied, system \eqref{eq:2} is permanent. 
\end{proof}

\begin{remark}
All the solutions of system \eqref{eq:2} are ultimately bounded above under the conditions given in Theorem \ref{thm:t2}. We can also prove that $\Gamma_\varepsilon \ne \phi$. That is, $\exists$ at least one bounded positive solution to system \eqref{eq:2}.
\end{remark} 

\section{Global Attractivity}\label{Sec:5}
\begin{theorem}\label{thm:t3}
Let $ (\chi^\ast(t), \Upsilon ^\ast(t)) $ (with $\chi^\ast>0$, $\Upsilon ^\ast>0$) be a bounded solution to system \eqref{eq:2}. Suppose the conditions $ b_L\varphi_L\xi_L > \nu_M \zeta_M $, $\kappa_M>m_L$, $(\kappa_L-m_M)\zeta_L k^{\varepsilon}_{\chi}> m_M(a_M+(1-\zeta_L)A_M)$, 
\begin{equation}\label{eq:3}
\begin{split}
\dfrac{-b(t)\varphi^2(t)e^{\varphi(t)\varsigma}}{(\varphi(t)e^{\varphi(t)\varsigma}+c(t)(e^{\varphi(t)\varsigma}-1)K^{\varepsilon}_{\chi})^2}+4c(t)k^{\varepsilon}_{\chi}+2\varphi(t) \\ +\dfrac{\nu(t) \zeta ^2(t)(k^{\varepsilon}_{\chi})^2}{((a(t)+\zeta(t)k^{\varepsilon}_{\chi}+(1-\zeta(t))A(t)+\xi(t)K^{\varepsilon}_{\Upsilon })^2} > 0
\end{split}
\end{equation}
and
\begin{equation}\label{eq:4}
\begin{split}
\dfrac{- \kappa(t) [\zeta ^2(t)(K^{\varepsilon}_{\chi})^2+a(t)A(t)(1-\zeta(t))+(1-\zeta(t))^2 A^2(t)]}{(a(t)+\zeta(t)K^{\varepsilon}_{\chi}+(1-\zeta(t))A(t)+\xi(t)k^{\varepsilon}_{\Upsilon })^2} + m(t) > 0
\end{split}
\end{equation}
hold, then $ (\chi^\ast(t), \Upsilon ^\ast(t)) $ is globally attractive.
\end{theorem}
\begin{proof}
Let us say that $(\chi(t), \Upsilon (t))$ satisfies system \eqref{eq:2}. Since $ \Gamma_\varepsilon $ is ultimately bounded region of this system, $ \exists $ $ T $ \textgreater $ 0 $ $\ni$ $ (\chi(t), \Upsilon (t))\in \Gamma_\varepsilon $ and $ (\chi^\ast(t), \Upsilon ^\ast(t))\in \Gamma_\varepsilon $ for all $ t\geqslant$ (initial time)$t_0$ $+T . $
Let us construct a function $H(t)$ such that 
\begin{center}
$ H(t)=\frac{1}{2}(\chi(t)-\chi^\ast(t))^2 + \frac{1}{2}(\Upsilon (t)- \Upsilon ^\ast(t))^2$.
\end{center}
For simplicity let 
\\
$ \varphi(t,t-\varsigma,\chi(t-\varsigma),\chi^\ast(t-\varsigma))= $ 
\begin{center}
$ [\varphi(t)e^{\varphi(t)\varsigma}+c(t)(e^{\varphi(t)\varsigma}-1)\chi(t-\varsigma)][\varphi(t)e^{\varphi(t)\varsigma}+c(t)(e^{\varphi(t)\varsigma}-1)\chi^\ast(t-\varsigma)] $ and
\end{center}
$ \varphi(t,\chi(t),\chi^\ast(t), \Upsilon (t), \Upsilon ^\ast(t))= $ 
\begin{center}
$ [a(t)+\zeta(t)\chi(t)+(1-\zeta(t))A(t)+\xi(t) \Upsilon (t)][a(t)+\zeta(t)\chi^\ast(t)+(1-\zeta(t))A(t)+\xi(t) \Upsilon ^\ast(t)] $.
\end{center}
Differentiation of $ H(t) $ w.r.t $t$ along system \eqref{eq:2} is
\begin{equation}\nonumber
\begin{split}
H’(t)=\dfrac{b(t)\varphi^2(t)e^{\varphi(t)\varsigma}}{\varphi(t,t-\varsigma,\chi(t-\varsigma),\chi^\ast(t-\varsigma))}[\chi(t)-\chi^\ast(t)][\chi(t-\varsigma)-\chi^\ast(t-\varsigma)] \vspace{0.1cm}\\
-\varphi(t)[\chi(t)-\chi^\ast(t)]^2 -c(t)[\chi(t)-\chi^\ast(t)]^2 [\chi(t)+\chi^\ast(t)]\\ -\dfrac{\nu(t)\zeta(t)[a(t)+(1-\zeta(t))A(t)]}{\varphi(t,\chi(t),\chi^\ast(t), \Upsilon (t), \Upsilon ^\ast(t))}[\chi(t)-\chi^\ast(t)][\chi(t) \Upsilon (t)-\chi^\ast(t) \Upsilon ^\ast(t)] \\
-\dfrac{\nu(t)p^2(t)\chi(t)\chi^\ast(t)}{\varphi(t,\chi(t),\chi^\ast(t), \Upsilon (t), \Upsilon ^\ast(t))}[\chi(t)-\chi^\ast(t)][ \Upsilon (t)- \Upsilon ^\ast(t)] \\
-\dfrac{\nu(t)\zeta(t)\xi(t) \Upsilon (t) \Upsilon ^\ast(t)}{\varphi(t,\chi(t),\chi^\ast(t), \Upsilon (t), \Upsilon ^\ast(t))}[\chi(t)-\chi^\ast(t)]^2 - m(t)[ \Upsilon (t)- \Upsilon ^\ast(t)]^2\\
+\dfrac{\kappa(t)\zeta(t)[a(t)+(1-\zeta(t))A(t)]}{\varphi(t,\chi(t),\chi^\ast(t), \Upsilon (t), \Upsilon ^\ast(t))}[\chi(t) \Upsilon (t)-\chi^\ast(t) \Upsilon ^\ast(t)][ \Upsilon (t)- \Upsilon ^\ast(t)]\\
+ \dfrac{\kappa(t)[\zeta ^2(t)\chi(t)\chi^\ast(t)+a(t)A(t)(1-\zeta(t))+(1-\zeta(t))^2 A^2(t)]}{\varphi(t,\chi(t),\chi^\ast(t), \Upsilon (t), \Upsilon ^\ast(t))}[ \Upsilon (t)- \Upsilon ^\ast(t)]^2\\
+\dfrac{\kappa(t)\zeta(t)\xi(t) \Upsilon (t) \Upsilon ^\ast(t)}{\varphi(t,\chi(t),\chi^\ast(t), \Upsilon (t), \Upsilon ^\ast(t))}[\chi(t)-\chi^\ast(t)][ \Upsilon (t)- \Upsilon ^\ast(t)]\\
+\dfrac{\kappa(t)\zeta(t)(1-\zeta(t))A(t)}{\varphi(t,\chi(t),\chi^\ast(t), \Upsilon (t), \Upsilon ^\ast(t))}[ \Upsilon (t)\chi^\ast(t)- \Upsilon ^\ast(t)\chi(t)][ \Upsilon (t)- \Upsilon ^\ast(t)].
\end{split}
\end{equation}
Using equations \eqref{eq:3}, \eqref{eq:4} and the inequality $xy\leqslant\frac{1}{2}(x^2+y^2)$ we get,
\begin{equation}\nonumber
\begin{split}
H’(t)\leqslant-\dfrac{1}{2}\Bigg[\frac{-b(t)\varphi^2(t)e^{\varphi(t)\varsigma}}{(\varphi(t)e^{\varphi(t)\varsigma}+c(t)(e^{\varphi(t)\varsigma}-1)K^{\varepsilon}_{\chi})^2}+4c(t)k^{\varepsilon}_{\chi}+2\varphi(t)\\
+\frac{\nu(t)p^2(t)(k^{\varepsilon}_{\chi})^2}{((a(t)+\zeta(t)k^{\varepsilon}_{\chi}+(1-\zeta(t))A(t)+\xi(t)K^{\varepsilon}_{\Upsilon })^2}\Bigg][\chi(t)-\chi^\ast(t)]^2\\
-\dfrac{1}{2}\Bigg[\frac{-2 \kappa(t) [\zeta ^2(t)(K^{\varepsilon}_{\chi})^2+a(t)A(t)(1-\zeta(t))+(1-\zeta(t))^2 A^2(t)]}{(a(t)+\zeta(t)K^{\varepsilon}_{\chi}+(1-\zeta(t))A(t)+\xi(t)k^{\varepsilon}_{\Upsilon })^2} \\
+ 2m(t)\Bigg][ \Upsilon (t)- \Upsilon ^\ast(t)]^2.
\end{split}
\end{equation}
Let 
\begin{center}
$M_1=$min$\Bigg[\dfrac{-b(t)\varphi^2(t)e^{\varphi(t)\varsigma}}{(\varphi(t)e^{\varphi(t)\varsigma}+c(t)(e^{\varphi(t)\varsigma}-1)K^{\varepsilon}_{\chi})^2}+4c(t)k^{\varepsilon}_{\chi}+2\varphi(t)$\\
$+\dfrac{\nu(t) \zeta ^2(t)(k^{\varepsilon}_{\chi})^2}{[a(t)+\zeta(t)k^{\varepsilon}_{\chi}+(1-\zeta(t))A(t)+\xi(t)K^{\varepsilon}_{\Upsilon }]^2}\Bigg]$, \vspace{0.2cm}\\
$M_2=$min$\Bigg[\dfrac{-2 \kappa(t) [\zeta ^2(t)(K^{\varepsilon}_{\chi})^2+a(t)A(t)(1-\zeta(t))+(1-\zeta(t))^2 A^2(t)]}{(a(t)+\zeta(t)K^{\varepsilon}_{\chi}+(1-\zeta(t))A(t)+\xi(t)k^{\varepsilon}_{\Upsilon })^2}+2m(t)\Bigg]$,
\end{center}
such that
\begin{center}
$ H’(t)\leqslant-\dfrac{1}{2}M_1[\chi(t)-\chi^\ast(t)]^2-\dfrac{1}{2}M_2[\Upsilon (t)- \Upsilon ^\ast(t)]^2 $.\\
\end{center}
For $M=$min$\{M_1,M_2\}$, we have
\begin{center}
$ H’(t)\leqslant-\dfrac{M}{2}\Bigg[[\chi(t)-\chi^\ast(t)]^2+[\Upsilon (t)- \Upsilon ^\ast(t)]^2\Bigg]$.\\
\end{center}
By the definition of $ H(t) $ we get,
\begin{center}
$ H’(t)\leqslant-MH(t) $. 
\end{center}
Integrating the above inequality we have
\begin{center}
$ \mbox{log}H(t)\leqslant-Mt $,
\end{center}
which is same as
\begin{center}
$ H(t)\leqslant e^{-Mt} $.
\end{center}
Since $ H(t) $ \textgreater $ 0 $, 
\begin{center}
$ \lim\limits_{t\rightarrow+\infty}H(t)=0 $.
\end{center}
This implies
\begin{center}
$ \lim\limits_{t\rightarrow+\infty}(|\chi(t)-\chi^\ast(t)|+| \Upsilon (t)- \Upsilon ^\ast(t)|)=0 $.
\end{center}
Hence the solution $(\chi^\ast(t), \Upsilon ^\ast(t)) $ is globally attractive. 
\end{proof}
\section{Existence of periodic solution}\label{Sec:6}
Owing to seasonality, factors like weather conditions, mating habits, availability of food, etc have an impact on the nature of every parameter of the system \eqref{eq:2}. These parameters can be periodic of some common period. Thus, incorporating a periodic environment to system \eqref{eq:2} we consider every parameter in system \eqref{eq:2} to be $ \rho $ periodic in $t$, \\
i.e, $ a(t+\rho)=a(t)$, $ b(t+\rho)=b(t)$, $ c(t+\rho)=c(t)$, $ m(t+\rho)=m(t) $, $ \zeta(t+\rho)=\zeta(t) $, $ A(t+\rho)=A(t) $, 
$\nu(t+\rho)=\nu(t) $, $\kappa(t+\rho)=\kappa(t)$, $ \varphi(t+\rho)=\varphi(t) $,
and $\xi(t+\rho)=\xi(t) $.\\

To derive the next theorem, we make use of the following details regarding the degree theory:
For normed vector spaces $X$ and $Z$ let us consider a linear operator $L$ that maps Domain $L$ to the normed vector space $Z$ where $X$ is a superset of Domain $L$. Consider an operator $S$ mapping from $X$ to $Z$ in such a way that $S$ is continuous. If the number of dissimilar members in Kernel $L$, called as dim of Kernel $L$, has the same value as that of co-dim of Image $L$ (a finite number $<$ infinity) and if $Z$ is a superset of the closed set Image $L$, then $L$ is called an index zero Fredholm mapping. If $B$ and $Q$ mapping from $X$ to itself and $Z$ to itself, respectively, are continuous linear transformations such that sets Image $B$ and Kernel $L$ are equal and sets Image $L$, Image $(I-Q)$ and Kernel $Q$ are equal. $L|$Domain $L$ $\cap$ Kernel $B$ mapping set $(I-B)X$ to set Image $L$ is said to be invertible if $L$ is an index zero Fredholm mapping and the inverse mapping of $L|$Domain $L$ $\cap$ Kernel $B$ is indicated by $K_B$. $S$ is called $L$-compact over $\Bar{\Sigma}$ if $K_B(I-Q)S$ mapping from $\Bar{\Sigma}$ to $X$ is compact and $QS(\Bar{\Sigma})$ is bounded, when $X$ is superset to $\Sigma$ and $\Sigma$ is open and bounded in it. There exists $J$ mapping from Image $Q$ to the set Kernel $L$, an isomorphism, if there is an isomorphism from Image $Q$ to Kernel $L$. 

\begin{lemma}\label{lem:1}
(Continuation theorem)\normalfont{\cite{GM77}} \textit{If $L$ is an index zero Fredholm mapping and $S$ is a L-compact set on $ \Bar{\Sigma} $, if constraints (i) and (ii) are satisfied then the equation $ Lx=Sx $ must have a minimum of one solution in the set Domain $L\cap\Bar{\Sigma}. $ \\
(i) Every solution $x$ of the equation $ Lx=\lambda Sx $ does not belong to the set $\partial\Sigma $ where $0 < \lambda < 1 $, \\
(ii) For every $ x$ lying in the set $\partial\Sigma\cap$ Kernel $L $, the values of $ deg\{JQS,\Sigma\cap$ Kernel $L,0\}$ (Brouwer degree) and $ QSx$ are not zeroes.}
\end{lemma}
For any continuous and $\rho $-periodic function $j(t)$,
$ \Tilde{j}=\frac{1}{\rho}\int_{0}^{\rho}j(t)dt $.
\begin{theorem}\label{thm:t4}
If $ \widetilde{b\varphi\xi}>\widetilde{\nu \zeta }$, $\widetilde{\kappa}>\widetilde{m}$ and $ \dfrac{(\widetilde{b\varphi\xi}-\widetilde{\nu \zeta })\zeta_L(\widetilde{\kappa}-\widetilde{m})}{\widetilde{\xi c}}exp\{-2\widetilde{\varphi}\rho\}>\widetilde{m}a_M $ hold then \vspace{0.2cm}\\system \eqref{eq:2} has at the minimum one positive solution that is $\rho$-periodic. 
\end{theorem}
\begin{proof}
Denote $ \chi(t)=\exp\{r(t)\} , \chi(t-\varsigma)=\exp\{r(t-\varsigma)\}, \Upsilon (t)=\exp\{s(t)\}$.\\
Substituting the above values in system \eqref{eq:2}, we get
\begin{equation}\label{eq:z}
\begin{split}
r’(t)=\frac{b(t)\varphi(t)e^{u(t-\varsigma)}}{e^{r(t)}[\varphi(t)e^{\varphi(t)\varsigma}+c(t)(e^{\varphi(t)\varsigma}-1)e^{u(t-\varsigma)}]}-\varphi(t)-c(t)e^{r(t)} \\ -\frac{\nu(t)\zeta(t)e^{s(t)}}{a(t)+\zeta(t)e^{r(t)}+(1-\zeta(t))A(t)+\xi(t)e^{s(t)}},\\
s’(t)=\dfrac{\kappa(t)\{\zeta(t)e^{r(t)}+(1-\zeta(t))A(t)\}}{a(t)+\zeta(t)e^{r(t)}+(1-\zeta(t))A(t)+\xi(t)e^{s(t)}}-m(t).
\end{split}
\end{equation}
In order to apply Lemma \ref{lem:2} to system \eqref{eq:2}, we take
\begin{center}
$ X=Z=\{(r,s)^T\in C(\mathbb{R},\mathbb{R}^2):r(t+\rho)=r(t),s(t+\rho)=s(t) $ $\forall$ $ t\in\mathbb{R}\} $.
\end{center}
Let
\begin{center}
$ \| (r,s) \|= $ $\max\limits_{t\in[0, \rho]}|r(t)|+\max\limits_{t\in[0, \rho]}|s(t)| $ for $ (r,s)\in X $(or $Z$).
\end{center}
Equipped with $ \| \, \| $ norm, $Z$ and $X$ are complete normed spaces. \\ Consider
\begin{center}
$ S \begin{pmatrix} r \\ s \end{pmatrix} $ = $\begin{pmatrix} S_1(t) \\ S_2(t) \end{pmatrix}$
\end{center}
\begin{center}
$ = \begin{pmatrix} \frac{b(t)\varphi(t)e^{r(t-\varsigma)}}{e^{r(t)}[\varphi(t)e^{\varphi(t)\varsigma}+c(t)(e^{\varphi(t)\varsigma}-1)e^{r(t-\varsigma)}]}-\varphi(t)-c(t)e^{r(t)} -\frac{\nu(t)\zeta(t)e^{s(t)}}{a(t)+\zeta(t)e^{r(t)}+(1-\zeta(t))A(t)+\xi(t)e^{s(t)}}\\
\frac{\kappa(t)\{\zeta(t)e^{r(t)}+(1-\zeta(t))A(t)\}}{a(t)+\zeta(t)e^{r(t)}+(1-\zeta(t))A(t)+\xi(t)e^{s(t)}}-m(t) \end{pmatrix}$,
\end{center}
\begin{center}
$ L \begin{pmatrix} r \\ s \end{pmatrix} = \begin{pmatrix} r' \\ s' \end{pmatrix} $ and 
$ B \begin{pmatrix} r \\ s \end{pmatrix} = \begin{pmatrix} \frac{1}{\rho}\int_{0}^{\rho}r(t)dt \\ \frac{1}{\rho}\int_{0}^{\rho}s(t)dt \end{pmatrix} 
= Q \begin{pmatrix} r \\ s \end{pmatrix} ; 
\begin{pmatrix} r \\ s \end{pmatrix} \in X $,
\end{center}
Then
\begin{center}
Ker $ L=\{(r,s)\in X/(r(t),s(t))\equiv (h_1,h_2)\in\mathbb{R}^2 $ for $ t\in\mathbb{R}\}, $ \vspace{0.1cm}\\
Im $ L=\{(r,s)\in Z/\int_{0}^{\rho}r(t)dt=0, \int_{0}^{\rho}s(t)dt=0 \}, $
\end{center}
where $ h_1, h_2 \in\mathbb{R}, $ dimension Kernel $ L$ and co-dimension Image $ L $ are equal to 2 and the set Image $L$ $ \in Z$ is closed. \\
Thus we write $L$ is an index zero Fredholm mapping. \\
$B$ and $Q$ are continuous projections $\ni$ Image $ B= $ Kernel $ L, $ Image $ L= $ Kernel $ Q= $ Image $ (I-Q). $ \\
$ K_B: $ Image $ L\rightarrow $ Domain $ L\cap $ Kernel $ B $ exists and is given by
\begin{center}
$ K_B \begin{pmatrix} r \\ s \end{pmatrix} = \begin{pmatrix} \int_{0}^{t}r(n)dn-\frac{1}{\rho }\int_{0}^{\rho}\int_{0}^{t}r(n)dn dt \\ \int_{0}^{t}s(n)dn-\frac{1}{\rho }\int_{0}^{\rho }\int_{0}^{t}s(n)dn dt \end{pmatrix}$.
\end{center}
We know that for any bounded and open set $ \Sigma $ that has the set $X$ as its superset, $S$ happens to be $L$-compact over $ \Bar{\Sigma} $ taking into account the continuity of sets $ QS $ and $ K_B(I-Q)S $. \\
For $0 < \lambda < 1 $, the equation $ Lx =\lambda Sx $ gives
\begin{equation}\nonumber
\begin{split}
r’(t)=\lambda\bigg[\frac{b(t)\varphi(t)e^{r(t-\varsigma)}}{e^{r(t)}[\varphi(t)e^{\varphi(t)\varsigma}+c(t)(e^{\varphi(t)\varsigma}-1)e^{r(t-\varsigma)}]}-\varphi(t)-c(t)e^{r(t)} \\ -\frac{\nu(t)\zeta(t)e^{s(t)}}{a(t)+\zeta(t)e^{r(t)}+(1-\zeta(t))A(t)+\xi(t)e^{s(t)}}\bigg],\vspace{0.2cm}\\
s’(t)=\lambda\bigg[\frac{\kappa(t)\{\zeta(t)e^{r(t)}+(1-\zeta(t))A(t)\}}{a(t)+\zeta(t)e^{r(t)}+(1-\zeta(t))A(t)+\xi(t)e^{s(t)}}-m(t)\bigg].
\end{split}
\end{equation}
For certain $\lambda$ $\ni$ $0 < \lambda < 1 $, if $ (r(t),s(t))\in X $ is any solution of the above equation, \\
we get, by integrating above equation over $ [0, \rho]$,
\begin{equation}\label{eq:a}
\begin{split}
\int_{0}^{\rho}\dfrac{b(t)\varphi(t)e^{t(t-\varsigma)}}{e^{r(t)}[\varphi(t)e^{\varphi(t)\varsigma}+c(t)(e^{\varphi(t)\varsigma}-1)e^{r(t-\varsigma)}]}dt=\int_{0}^{\rho}\varphi(t)dt+\int_{0}^{\rho}c(t)e^{r(t)}dt \\
+\int_{0}^{\rho}\dfrac{\nu(t)\zeta(t)e^{s(t)}}{a(t)+\zeta(t)e^{r(t)}+(1-\zeta(t))A(t)+\xi(t)e^{s(t)}}dt,\\
\int_{0}^{\rho}\dfrac{\kappa(t)\{\zeta(t)e^{r(t)}+(1-\zeta(t))A(t)\}}{a(t)+\zeta(t)e^{r(t)}+(1-\zeta(t))A(t)+\xi(t)e^{s(t)}}dt=\int_{0}^{\rho}m(t)dt.
\end{split}
\end{equation}
Hence we have 
\begin{center}
$ \int_{0}^{\rho}|r’(t)|dt\leqslant\lambda\int_{0}^{\rho}\dfrac{b(t)\varphi(t)e^{r(t-\varsigma)}}{e^{r(t)}[\varphi(t)e^{\varphi(t)\varsigma}+c(t)(e^{\varphi(t)\varsigma}-1)e^{r(t-\varsigma)}]}dt+\lambda\int_{0}^{\rho}\varphi(t)dt-\lambda\int_{0}^{\rho}c(t)e^{r(t)}dt $ \\
$ -\lambda\int_{0}^{\rho}\dfrac{\nu(t)\zeta(t)e^{s(t)}}{a(t)+\zeta(t)e^{r(t)}+(1-\zeta(t))A(t)+\xi(t)e^{s(t)}}dt $,\\
$ \int_{0}^{\rho}|s’(t)|dt\leqslant\lambda\int_{0}^{\rho}\dfrac{\kappa(t)\{\zeta(t)e^{r(t)}+(1-\zeta(t))A(t)\}}{a(t)+\zeta(t)e^{r(t)}+(1-\zeta(t))A(t)+\xi(t)e^{s(t)}}dt+\lambda\int_{0}^{\rho}m(t)dt.$
\end{center}
Using equation \eqref{eq:a}, for $\lambda\in (0,1)$, we get
\begin{center}
$ \int_{0}^{\rho} |r’(t)| dt \leqslant 2 \lambda \int_{0}^{\rho} \varphi(t)dt<2 \int_{0}^{\rho}\varphi(t)dt $, \\
$ \int_{0}^{\rho} |s’(t) |dt \leqslant 2 \lambda \int_{0}^{\rho} m(t)dt<2 \int_{0}^{\rho}m(t)dt $.
\end{center}
Hence, we can say that
\begin{equation}\label{eq:b}
\int_{0}^{\rho}|r’(t)|dt\leqslant2\widetilde{\varphi}\rho, \, \, \int_{0}^{\rho}|s’(t)|dt\leqslant2\widetilde{m}\rho.
\end{equation}
We now denote
\begin{equation}\label{eq:d}
r(\psi_1)=\min\limits_{t\in[0, \rho]}r(t), \, \, r(\psi_2)=\max\limits_{t\in[0, \rho]}r(t), \, \, s(\psi_1)=\min\limits_{t\in[0, \rho]}s(t), \, \, s(\psi_2)=\max\limits_{t\in[0, \rho]}s(t).
\end{equation}
From equations \eqref{eq:a}, \eqref{eq:d} we get
\begin{center}
$ \bigint_{0}^{\rho}\dfrac{b(t)\varphi(t)}{e^{r(\psi_1)}c(t)(e^{\varphi(t)\varsigma}-1)}dt \geqslant \int_{0}^{\rho}c(t)e^{r(\psi_1)}dt= \widetilde{c}\rho e^{r(\psi_1)} $.
\end{center}
That is,
\begin{equation}\label{eq:c}
r(\psi_1) \leqslant \dfrac{1}{2}\mbox{log}\Bigg[{\dfrac{\widetilde{b\varphi}}{(\widetilde{c})^2}}\Bigg]=u_1.
\end{equation} 
Adding equations \eqref{eq:b}, \eqref{eq:c} we get
\begin{center}
$ r(t)\leqslant r(\psi_1)+\int_{0}^{\rho}|r’(t)|dt\leqslant \dfrac{1}{2}\mbox{log}\Bigg[{\dfrac{\widetilde{b\varphi}}{(\widetilde{c})^2}}\Bigg] +2\widetilde{\varphi}\rho=H_1 $.
\end{center}
Hence we obtain
\begin{center}
$ r(t)\leqslant H_1 $.
\end{center}
From equations \eqref{eq:a}, \eqref{eq:d} we get
\begin{center}
$ \bigint_{0}^{\rho}\dfrac{b(t)\varphi(t)}{c(t)e^{r(\psi_2)}}dt\leqslant\int_{0}^{\rho}c(t)e^{r(\psi_2)}dt+\int_{0}^{\rho}\dfrac{\nu(t)\zeta(t)}{\xi(t)}dt = e^{r(\psi_2)} \rho\widetilde{c}+\dfrac{\widetilde{\nu \zeta }}{\widetilde{\xi}}\rho $.
\end{center}
That is,
\begin{equation}\label{eq:e}
r(\psi_2) \geqslant \mbox{log}\Bigg[\dfrac{1}{\widetilde{c}}(\widetilde{b\varphi}-\dfrac{\widetilde{\nu \zeta}}{\widetilde{\xi}})\Bigg]=U_1.
\end{equation}
Subtracting equation \eqref{eq:b} from equation \eqref{eq:e} we get
\begin{center}
$ r(t)\geqslant r(\psi_2)-\int_{0}^{\rho}|r’(t)|dt\geqslant \mbox{log}\Bigg[\dfrac{1}{\widetilde{c}}(\widetilde{b\varphi}-\dfrac{\widetilde{\nu \zeta }}{\widetilde{\xi}})\Bigg]-2\widetilde{\varphi}\rho=H_2 $.
\end{center}
Hence we obtain
\begin{center}
$ r(t)\geqslant H_2 $.
\end{center}
Therefore $\max\limits_{t\in[0,\rho]}|r(t)|\leqslant $ max $ \{|H_1|,|H_2|\}=A_1 $ ($ A_1 $ is independent of $ \lambda $). \\
From equations \eqref{eq:a}, \eqref{eq:d} we get
\begin{center}
$ \widetilde{m}\rho \leqslant \bigint_{0}^{\rho} \, \dfrac{\zeta_L \kappa(t) e^{2\widetilde{\varphi}\rho}\Bigg[\dfrac{\widetilde{b\varphi}}{(\widetilde{c})^2}\Bigg]^{(1/2)}}{e^{2\widetilde{\varphi}\rho}\zeta_L\Bigg[\dfrac{\widetilde{b\varphi}}{(\widetilde{c})^2}\Bigg]^{(1/2)} +e^{s(\psi_1)} \xi_L}dt = \dfrac{e^{2\widetilde{\varphi}\rho}\Bigg[\dfrac{\widetilde{b\varphi}}{(\widetilde{c})^2}\Bigg]^{(1/2)}}{e^{2\widetilde{\varphi}\rho}\zeta_L\Bigg[\dfrac{\widetilde{b\varphi}}{(\widetilde{c})^2}\Bigg]^{(1/2)} +e^{s(\psi_1)} \xi_L}[\rho \zeta_L \widetilde{\kappa}] $.
\end{center}
That is,
\begin{equation}\label{eq:f}
s(\psi_1) \leqslant \mbox{log}\Bigg[\dfrac{[\widetilde{b\varphi}]^{(1/2)} e^{2\widetilde{\varphi}\rho}(\widetilde{\kappa}-\widetilde{m})\zeta_L}{\widetilde{mc}\,\xi_L}\Bigg]=u_2.
\end{equation} 
Adding equations \eqref{eq:b}, \eqref{eq:f} we get
\begin{center}
$ s(t)\leqslant s(\psi_1)+\int_{0}^{\rho}|s’(t)|dt\leqslant \mbox{log}\Bigg[\dfrac{[\widetilde{b\varphi}]^{(1/2)} e^{2\widetilde{\varphi}\rho}(\widetilde{\kappa}-\widetilde{m})\zeta_L}{\widetilde{mc}\,\xi_L}\Bigg]+2\widetilde{m}\rho=H_3 $.
\end{center}
Hence we obtain
\begin{center}
$ s(t)\leqslant H_3 $.
\end{center}
From equations \eqref{eq:a}, \eqref{eq:d} we get
\begin{center}
$ \widetilde{m}\rho \geqslant \bigint_{0}^{\rho}\dfrac{\kappa(t)\zeta(t)\Bigg[\dfrac{1}{\widetilde{c}}(\widetilde{b\varphi}-\dfrac{\widetilde{\nu \zeta}}{\widetilde{\xi}})\Bigg]}{a_M (2\widetilde{\varphi}\rho) +\zeta_L \Bigg[\dfrac{1}{\widetilde{c}}(\widetilde{b\varphi}-\dfrac{\widetilde{\nu \zeta}}{\widetilde{\xi}})\Bigg] +(e^{s(\psi_2)})\xi_L} dt$.
\end{center}
That is,
\begin{equation}\label{eq:g}
s(\psi_2) \geqslant \mbox{log}\Bigg[\dfrac{1}{\widetilde{m}\xi_L}\Bigg[-\widetilde{m}(2\widetilde{\varphi}\rho)a_M+\Bigg[\dfrac{1}{\widetilde{c}}(\widetilde{b\varphi}-\dfrac{\widetilde{\nu \zeta}}{\widetilde{\xi}})\Bigg] [\widetilde{\kappa}-\widetilde{m}]\zeta_L\Bigg] \Bigg]=U_2.
\end{equation}
Subtracting equation \eqref{eq:b} from equation \eqref{eq:g} we get
\begin{center}
$ s(t)\geqslant s(\psi_2)-\int_{0}^{\rho}|s’(t)|dt\geqslant \mbox{log}\Bigg[\dfrac{1}{\widetilde{m}\xi_L}\Bigg[-\widetilde{m}(2\widetilde{\varphi}\rho)a_M+\Bigg[\dfrac{1}{\widetilde{c}}(\widetilde{b\varphi}-\dfrac{\widetilde{\nu \zeta}}{\widetilde{\xi}})\Bigg] [\widetilde{\kappa}-\widetilde{m}]\zeta_L\Bigg] \Bigg] -2\widetilde{m}\rho=H_4 $.
\end{center}
Hence we obtain
\begin{center}
$ s(t)\geqslant H_4 $.
\end{center}
Therefore $\max\limits_{t\in[0,\rho]}|s(t)|\leqslant $ max $ \{|H_3|,|H_4|\}=A_2 $ ($ A_2 $ is independent of $ \lambda $). \\
Take $ A_4=A_1+A_2+A_3 $, where $ A_3>0 $ such that $ A_3>|u_1|+|u_2|+|U_1|+|U_2| $. \\
Now let us see the following equations,
\begin{center}
$ \vspace{0.5cm}\bigint_{0}^{\rho}\dfrac{b(t)\varphi(t)e^{r(t-\varsigma)}}{e^{r}[\varphi(t)e^{\varphi(t)\varsigma}+c(t)(e^{\varphi(t)\varsigma}-1)e^{r(t-\varsigma)}]}dt-\widetilde{\varphi}-e^{r}\widetilde{c}-\dfrac{1}{\rho} \bigint_{0}^{\rho}\dfrac{\mu \nu(t)\zeta(t)e^{s}}{a(t)+\zeta(t)e^{r}+(1-\zeta(t))A(t)+\xi(t)e^{s}}dt=0 $, \\
$-\widetilde{m}+\dfrac{1}{\rho}\bigint_{0}^{\rho}\dfrac{\kappa(t)\{\zeta(t)e^{r}+(1-\zeta(t))A(t)\}}{a(t)+\zeta(t)e^{r}+(1-\zeta(t))A(t)+\xi(t)e^{s}}dt=0$
\end{center}
for $(r,s) \in \mathbb{R}^2 $, wherein $ \mu\in [0,1] $.
We can prove that every solution $ (r^\ast,s^\ast) $ of the above equations satisfies
\begin{equation}\label{eq:h}
u_1\leqslant r^\ast\leqslant U_1, u_2\leqslant s^\ast\leqslant U_2 
\end{equation}
by following similar procedure as above. Let 
\begin{center}
$ \Sigma=\{(r,s)^T\in X/ $ $ \|(r,s)\|< A_4\}. $ 
\end{center}
We can easily show that the set $ \Sigma $ fits into constraint (i) of Lemma \ref{lem:2}. \\
When $ (r,s)\in\partial\Sigma$ $\cap $ Kernal $ L=\partial\Sigma\cap\mathbb{R}^2, $ $(r,s)$ is a stationary vector in $\mathbb{R}^2$ such that $ \|(r,s)\|=|r|+|s|=A_4 $. \\
From the definition of $A_4$ and equation \eqref{eq:h}, we have
\begin{center}
$ QN \begin{pmatrix} r \\ s \end{pmatrix} $ \\
= $ \begin{pmatrix} \bigint_{0}^{\rho}\frac{b(t)\varphi(t)e^{r(t-\varsigma)}}{e^{r}[\varphi(t)e^{\varphi(t)\varsigma}+c(t)(e^{\varphi(t)\varsigma}-1)e^{r(t-\varsigma)}]}dt-\widetilde{\varphi}-e^{r}\widetilde{c} - \frac{1}{\rho} \bigint_{0}^{\rho}\frac{ \nu(t)\zeta(t)e^{s}}{a(t)+\zeta(t)e^{r}+(1-\zeta(t))A(t)+\xi(t)e^{s}}dt \\
-\widetilde{m}+\frac{1}{\rho}\bigint_{0}^{\rho}\frac{\kappa(t)\{\zeta(t)e^{r}+(1-\zeta(t))A(t)\}}{a(t)+\zeta(t)e^{r}+(1-\zeta(t))A(t)+\xi(t)e^{s}}dt \end{pmatrix} $ \vspace{0.1cm}\\
$\ne\begin{pmatrix} 0 \\ 0 \end{pmatrix}$.
\end{center}
That is the first part of condition (ii) of Lemma \ref{lem:2} is valid. \\
For any $\mu\in [0,1] $ we write a homotopy as given below
\begin{center}
$ H_\mu((r,s)^T)=\mu QS((r,s)^T)+(1-\mu)G((r,s)^T), $
\end{center}
where
\begin{center}
$ G((r,s)^T)= \begin{pmatrix} \bigint_{0}^{\rho}\dfrac{b(t)\varphi(t)e^{r(t-\varsigma)}}{e^{r}[\varphi(t)e^{\varphi(t)\varsigma}+c(t)(e^{\varphi(t)\varsigma}-1)e^{r(t-\varsigma)}]}dt-\widetilde{\varphi}-e^{r}\widetilde{c}\\ \widetilde{m}-\dfrac{1}{\rho}\bigint_{0}^{\rho}\dfrac{\kappa(t)\{\zeta(t)e^{r}+(1-\zeta(t))A(t)\}}{a(t)+\zeta(t)e^{r}+(1-\zeta(t))A(t)+\xi(t)e^{s}}dt \end{pmatrix}$. \vspace{0.4cm}
\end{center}
From equation \eqref{eq:h} it follows that $ 0\notin H_\mu(\partial\Sigma $ $\cap $ Ker $ L) $ for $ 0 \leqslant \mu \leqslant 1. $ 
$ G((r,s)^T)=0$ equation posesses a solution in $ \mathbb{R}^2 that is unique. $
Since the sets Image $Q$ and Kernel $L$ are equal, $ J=I $. Making use of homotopy invariance property, 
\begin{center}
$\deg\{JQS,\Sigma $ $ \cap $ Kernel $ L,0\}=\deg\{QS,\Sigma $ $ \cap $ Kernel $ L,0\}=\deg\{G,\Sigma $ $ \cap $ Kernel $ L,0\}\ne0 $.
\end{center}
Therefore from Lemma \ref{lem:2}, we know that $ Lx=Sx $ has atleast one solution lying in Domain $ L\cap\Bar{\Sigma}. $
\, That is, equation \eqref{eq:z} has atleast one $\rho$-periodic solution in Domain $ L\cap\Bar{\Sigma} $, say $ (r^\ast(t),s^\ast(t)). $ \\
Set $ \chi^\ast(t)=\exp\{r^\ast(t)\}, $ $ \Upsilon ^\ast(t)=\exp\{s^\ast(t)\} $ and $ (\chi^\ast(t), \Upsilon ^\ast(t)) $ is an $\rho$-periodic solution of system \eqref{eq:2} with strictly positive components. Hence proved.
\end{proof}
\textbf{Remark:} We just showed that $ (\chi^\ast(t), \Upsilon ^\ast(t))$ is an $ \rho $ periodic solution of system \eqref{eq:2} where $ \chi^\ast(t)=\exp\{r^\ast(t)\} $ and $ \Upsilon ^\ast(t)=\exp\{s^\ast(t)\} $. The values of $r^\ast$ and $ s^\ast$ are such that $ u_1 \leqslant r^\ast \leqslant U_1 $ and $ u_2 \leqslant s^\ast \leqslant U_2$ where \vspace{0.2cm}\\
$ u_1 = \dfrac{1}{2}\mbox{log}\Bigg[{\dfrac{\widetilde{b\varphi}}{(\widetilde{c})^2}}\Bigg]$, $ u_2 = \mbox{log}\Bigg[\dfrac{[\widetilde{b\varphi}]^{(1/2)} e^{2\widetilde{\varphi}\rho}(\widetilde{\kappa}-\widetilde{m})\zeta_L}{\widetilde{mc}\,\xi_L}\Bigg] $, \vspace{0.2cm}\\
$ U_1 = \mbox{log}\Bigg[\dfrac{1}{\widetilde{c}}(\widetilde{b\varphi}-\dfrac{\widetilde{\nu \zeta }}{\widetilde{\xi}})\Bigg] $ and $ U_2 = \mbox{log}\Bigg[\dfrac{1}{\widetilde{m}\xi_L}\Bigg[-\widetilde{m}(2\widetilde{\varphi}\rho)a_M+\Bigg[\dfrac{1}{\widetilde{c}}(\widetilde{b\varphi}-\dfrac{\widetilde{\nu \zeta }}{\widetilde{\xi}})\Bigg] [\widetilde{\kappa}-\widetilde{m}]\zeta_L\Bigg] \Bigg] $. \vspace{0.2cm}\\
\textbf{Remark:} One can show that $b_L \varphi_L \xi_L>\nu_M \zeta_M$ implies $\widetilde{b\varphi\xi}>\widetilde{\nu \zeta }$, $\kappa_M>m_L$ implies $\widetilde{\kappa}>\widetilde{m}$ and $(\kappa_L-m_M)\zeta_L k^{\varepsilon}_{\chi}> m_M(a_M+(1-\zeta_L)A_M)$ implies $\dfrac{(\widetilde{b\varphi\xi}-\widetilde{\nu \zeta })\zeta_L(\widetilde{\kappa}-\widetilde{m})}{\widetilde{\xi c}}exp\{-2\widetilde{\varphi}\rho\}>\widetilde{m}a_M $ if $\varepsilon$ is chosen appropriately, ie,\\
\begin{center}
$\dfrac{b_L \varphi_L \xi_L-\nu_M \zeta_M}{c_M \xi_L}(1-exp\{-2\widetilde{\varphi}\rho\})<\varepsilon<\dfrac{b_L \varphi_L \xi_L-\nu_M \zeta_M}{c_M \xi_L}$.
\end{center}
\section{Existence of almost periodic solution}\label{sec:7}
The perturbations in the prey-predator interactions are not always periodic. For rationally independent periods, the perturbations caused are not periodic, they are said to be quasi periodic or almost periodic \cite{FK04}. Let the functions $\varphi(t)$, $\nu(t)$, $c(t)$, $m(t)$, $b(t)$, $\zeta(t)$, $A(t)$, $\kappa(t)$, $a(t)$ and $\xi(t)$ be almost periodic in $t$.
\begin{lemma}\label{lem:2}
(Arzela-Ascoli theorem)\normalfont{\cite{BB64}} \textit{ Suppose that $f$, $g$ are any two positive integers, $K$ $\subset$ $\mathbb{R}^f$ and $K$ is compact in it and $\sigma=\{n/ n \in C(K,\mathbb{R}^g)\}$, then below mentioned properties are identical: \\
(i) $\sigma$ is a bounded set and is equi-continuous on $K$. \\
(ii) Each sequence in $\sigma$ contains a subsequence that becomes convergent in $K$, uniformly.}
\end{lemma}
\begin{theorem}\label{thm:t5}
If all the conditions in Theorem \ref{thm:t3} hold, there is a unique solution to system \eqref{eq:2} that is almost periodic in nature.
\end{theorem}
\begin{proof}
Every solution of system \eqref{eq:2} will be ultimately bounded above under the constraints given in Theorem \ref{thm:t2}. Therefore we talk about a bounded (by positive constants) positive($>0$) solution $y(t)=(y_1(t),y_2(t))$ to system \eqref{eq:2}.\vspace{0.3cm}\\ 
Thus $\exists$ a sequence ${t_j}$, $t_j \rightarrow$ infinity as $j \rightarrow$ infinity $\ni$ $(y_1(t+t_j),y_2(t+t_j))^T$ satisfies \vspace{0.2cm}
\begin{center}
$\vspace{0.3cm}\chi’(t)=\dfrac{b(t+t_j)\varphi(t+t_j)\chi(t-\varsigma)}{\varphi(t+t_j)e^{\varphi(t+t_j)\varsigma}+c(t+t_j)(e^{\varphi(t+t_j)\varsigma}-1)\chi(t-\varsigma)}-\varphi(t+t_j)\chi(t)-c(t+t_j)\chi^2(t) $ \vspace{0.1cm}\\
$-\dfrac{\nu(t+t_j)\zeta(t+t_j)\chi(t) \Upsilon (t)}{a(t+t_j)+\zeta(t+t_j)\chi(t)+(1-\zeta(t+t_j))A(t+t_j)+\xi(t+t_j) \Upsilon (t)}$ \vspace{0.3cm}\\
$\Upsilon '(t)=\dfrac{\kappa(t+t_j) \{\zeta(t+t_j)\chi(t)+(1-\zeta(t+t_j))A(t+t_j)\} \Upsilon (t)}{a(t+t_j)+\zeta(t+t_j)\chi(t)+(1-\zeta(t+t_j))A(t+t_j)+\xi(t+t_j) \Upsilon (t)}-m(t+t_j) \Upsilon (t)\cdot$ \vspace{0.2cm}
\end{center}
Therefore ${y_i (t+t_j)}$ for $i=1,2$ and ${\Dot{y_i}(t+t_j)}$ for $i=1,2$ are bounded uniformly and equi-continuous. Using Lemma \ref{lem:2} we can say that $\exists$ a sub-sequence ${y_i (t+t_m)}\subset{y_i (t+t_j)}$ which is uniformly convergent and for any $\varepsilon>0$ $\exists$ $\eta>0$ such that $|y_i (t+t_k)-y_i (t+t_m)|<\eta$, $i=1,2$ if $k,m>\eta>0$. \\
Hence $y_i (t)$ for $i=1,2$ is almost periodic as $t \rightarrow \infty$ (asymptotically). \\ Therefore ${y_i (t+t_m)}$ can be expressed as
\begin{center}
$y_i (t+t_m)= y_{i1}(t+t_m)+y_{i2}(t+t_m)$ 
\end{center}
where $y_{i1}(t+t_m)$ and $y_{i2}(t+t_m)$ are almost periodic and continuous functions respectively, for all $t\in\mathbb{R}$.\\ Further
\begin{center}
$\lim\limits_{m\rightarrow\infty}y_{i2}(t+t_m)=0$, \vspace{0.2cm} \\ $\lim\limits_{m\rightarrow\infty}y_{i1}(t+t_m)=y_{i1}(t)$ 
\end{center}
where function $y_{i1}$ is almost periodic. \vspace{0.1cm}\\ Hence
\begin{center}
$\lim\limits_{m\rightarrow\infty}
y_i (t+t_m)=y_{i1}(t)$ ($i=1,2$).
\end{center}
Also, 
\begin{align*}
\vspace{0.2cm}\lim\limits_{m\rightarrow\infty}\Dot{y_i}(t+t_m)=\lim\limits_{m\rightarrow\infty}\lim\limits_{h\rightarrow 0}\dfrac{y_i(t+t_m+h)-y_i(t+t_m)}{h}\vspace{0.4cm}\\
=\lim\limits_{h\rightarrow0}\lim\limits_{m\rightarrow\infty}\dfrac{y_i(t+t_m+h)-y_i(t+t_m)}{h}
\end{align*}
\begin{align*}
=\lim\limits_{h\rightarrow0}\dfrac{y_{i1}(t+h)-y_{i1}(t)}{h}\cdot
\end{align*}
Thus $\Dot{y_{i1}}$ exists for $i=1,2$. Thus, a sequence ${t_j}$ exists $\ni$ the value $t_j\rightarrow\infty$ as $j\rightarrow\infty$. Due to this, $a(t+t_j)\rightarrow a(t)$, $b(t+t_j)\rightarrow b(t)$, $c(t+t_j)\rightarrow c(t)$, $m(t+t_j)\rightarrow m(t)$, \vspace{0.1cm}\\ $\zeta(t+t_j)\rightarrow \zeta(t)$, $A(t+t_j)\rightarrow A(t)$, 
$\nu(t+t_j)\rightarrow \nu(t)$,
$\kappa(t+t_j)\rightarrow \kappa(t)$, $\varphi(t+t_j)\rightarrow\varphi(t)$ and \vspace{0.1cm}\\
$\xi(t+t_j)\rightarrow \xi(t)$. \vspace{0.5cm} \\
$\Dot{y_{11}}=\lim\limits_{j \rightarrow\infty}\dfrac{d}{dt}y_1 (t+t_j)$\\
\begin{align*}
\vspace{0.5cm}= \lim\limits_{j \rightarrow\infty}\bigg[\dfrac{b(t+t_j)\varphi(t+t_j)y_1 (t+t_j-\varsigma)} {\varphi(t+t_j)e^{\varphi(t+t_j)\varsigma}+c(t+t_j)(e^{\varphi(t+t_j)\varsigma}-1)y_1 (t+t_j-\varsigma_1)}
\vspace{1cm}\\ -\varphi(t+t_j)y_1 (t+t_j) -c(t+t_j){y_1}^2(t+t_j)\\ \vspace{1cm} -\dfrac{\nu(t+t_j)\zeta(t+t_j)y_1 (t+t_j)y_2 (t+t_j)}{a(t+t_j)+\zeta(t+t_j)y_1 (t+t_j)+(1-\zeta(t+t_j))A(t+t_j)+\xi(t+t_j)y_2 (t+t_j)}\bigg]\cdot \vspace{1cm}
\end{align*}\\
We get
\begin{flalign*}
\Dot{y_{11}}=\dfrac{b(t)\varphi(t)y_1 (t-\varsigma)}{\varphi(t)e^{\varphi(t)\varsigma}+c(t)(e^{\varphi(t)\varsigma}-1)y_1 (t-\varsigma)}-\varphi(t)y_1 (t)-c(t){y_1}^2(t)\\
-\dfrac{\nu(t)\zeta(t)y_1 (t)y_2 (t)}{a(t)+ \zeta(t)y_1 (t)+(1-\zeta(t))A(t)+\xi(t)y_2 (t)}\cdot
\end{flalign*}
$\vspace{1cm}\Dot{y_{21}}=\lim\limits_{j\rightarrow\infty}\dfrac{d}{dt}y_2 (t+t_j)$ \vspace{-1cm}
\begin{align*}
=\lim\limits_{j\rightarrow\infty}\bigg[\dfrac{\kappa(t+t_j)\{\zeta(t+t_j)y_1 (t+t_j)+(1-\zeta(t+t_j))A(t+t_j)\}y_2 (t+t_j)}{a(t+t_j)+\zeta(t+t_j)y_1 (t+t_j)+(1-\zeta(t+t_j))A(t+t_j)+\xi(t+t_j)y_2 (t+t_j)}\\ 
-m(t+t_j)y_2 (t+t_j)\bigg]\cdot \vspace{0.1cm}
\end{align*}
We get 
\begin{flalign*}
\Dot{y_{21}}=\dfrac{\kappa(t)\{ \zeta(t)y_1 (t) +(1-\zeta(t))A(t)\}y_2 (t)}{a(t)+\zeta(t)y_1 (t) +(1-\zeta(t))A(t)+\xi(t)y_2 (t)} -m(t)y_2(t)\cdot
\end{flalign*}
Hence $(y_{11},y_{21})$ satisfies system \eqref{eq:2} and it is almost periodic. Thus, system \eqref{eq:2} possesses a unique positive solution that is almost periodic.
\end{proof}

\begin{remark}
Results proved in Theorems \ref{thm:t1}, \ref{thm:t2} and \ref{thm:t3} for positive invariance, permanence and global attractivity remain valid for system \eqref{eq:2} with almost periodic coefficients.
\end{remark}

\section{Delayed harvesting in prey}\label{sec:8}
As mentioned in Section \ref{Sec:1}, for biological and economic benefits, we must harvest hilsa when it crosses its maturity age and harvestable yield. This way we can assure that the population will be healthy and sustainable. Such a scenario has been modeled into system \eqref{eq:k} where $\varsigma_1$ is the delay in maturity and $\varsigma_2$ is the delay in harvest. Dynamical aspects of system \eqref{eq:k} are studied in Theorems \ref{thm:t6}-\ref{thm:t10} (proofs are as that of Theorems \ref{thm:t1}-\ref{thm:t5}). 

\begin{theorem}\label{thm:t6}
If $ b_L\varphi_L\xi_L > \nu_M \zeta_M $, $\kappa_M>m_L$ and $(\kappa_L-m_M)\zeta_L k^{\varepsilon}_{\chi}> m_M(a_M+(1-\zeta_L)A_M)$ then the set defined by 
\begin{center}
$ \Gamma_\varepsilon=\{(\chi(t), \Upsilon (t))\in \mathbb{R}^2 / k^{\varepsilon}_{\chi} \leqslant \chi(t) \leqslant K^{\varepsilon}_{\chi}, k^{\varepsilon}_{\Upsilon } \leqslant \Upsilon (t) \leqslant K^{\varepsilon}_{\Upsilon } \} $
\end{center}
is positively invariant 
w.r.t system \eqref{eq:k} where
$ K^{\varepsilon}_{\chi}$ ,
$ k^{\varepsilon}_{\chi}$,
$K^{\varepsilon}_{\Upsilon }$ and $k^{\varepsilon}_{\Upsilon }$ are as defined in Theorem \ref{thm:t1} and $ \varepsilon \geqslant 0 $ is sufficiently small so that $ k^{\varepsilon}_{\chi} > 0.$
\end{theorem}
\begin{theorem}\label{thm:t7}
If $ b_L\varphi_L\xi_L > \nu_M \zeta_M $, $\kappa_M>m_L$ and $(\kappa_L-m_M)\zeta_L k^{0}_{\chi}> m_M(a_M+(1-\zeta_L)A_M)$\vspace{0.2cm}\\ hold then system \eqref{eq:k} is permanent, where $k^{0}_{\chi}=\dfrac{b_L\varphi_L\xi_L-\nu_M \zeta_M }{c_M \xi_L}\cdot$
\end{theorem}
\begin{theorem}\label{thm:t8}
If $ b_L\varphi_L\xi_L > \nu_M \zeta_M $, $\kappa_M>m_L$, $(\kappa_L-m_M)\zeta_L k^{\varepsilon}_{\chi}> m_M(a_M+(1-\zeta_L)A_M)$,
\begin{equation}\nonumber
\begin{split}
\frac{-b(t)\varphi^2(t)e^{\varphi(t)\varsigma_1}}{(\varphi(t)e^{\varphi(t)\varsigma_1}+c(t)(e^{\varphi(t)\varsigma_1}-1)K^{\varepsilon}_{\chi})^2}+2\varphi(t)
+\frac{\nu(t) \zeta ^2(t)(k^{\varepsilon}_{\chi})^2}{((a(t)+\zeta(t)k^{\varepsilon}_{\chi}+(1-\zeta(t))A(t)+\xi(t)K^{\varepsilon}_{\Upsilon })^2} \\
+4c(t)k^{\varepsilon}_{\chi}+\frac{q(t)E(t)\varphi^2(t)e^{\varphi(t)\varsigma_2}}{(\varphi(t)e^{\varphi(t)\varsigma_2}+c(t)(e^{\varphi(t)\varsigma_2}-1)k^{\varepsilon}_{\chi})^2} > 0
\end{split}
\end{equation}
and
\begin{equation}\nonumber
\begin{split}
\frac{-\kappa(t)[\zeta ^2(t)(K^{\varepsilon}_{\chi})^2+a(t)A(t)(1-\zeta(t))+(1-\zeta(t))^2 A^2(t)]}{(a(t)+\zeta(t)K^{\varepsilon}_{\chi}+(1-\zeta(t))A(t)+\xi(t)k^{\varepsilon}_{\Upsilon })^2} + m(t) > 0
\end{split}
\end{equation}
hold, then, a bounded positive solution of system \eqref{eq:k},
$ (\chi^\ast(t), \Upsilon ^\ast(t)) $ is globally attractive.
\end{theorem}
\begin{theorem}\label{thm:t9}
If $ \widetilde{b\varphi\xi}>\widetilde{\nu \zeta }$, $\widetilde{\kappa}>\widetilde{m}$ and $ \dfrac{(\widetilde{b\varphi\xi}-\widetilde{\nu \zeta })\zeta_L(\widetilde{\kappa}-\widetilde{m})}{\widetilde{\xi c}}exp\{-2\widetilde{\varphi}\rho\}>\widetilde{m}a_M $ hold then \vspace{0.2cm}\\system \eqref{eq:k} has a minimum of one positive $\rho$-periodic solution. 
\end{theorem}
\begin{theorem}\label{thm:t10}
If all the conditions in Theorem \ref{thm:t9} are fulfilled, there is a positive solution to system \eqref{eq:k} that is distinct and almost periodic.
\end{theorem}

\section{Numerical simulation}\label{sec:9}
Behaviour of systems \eqref{eq:2} and \eqref{eq:k} has been tested by computer simulation. Various dynamical aspects like positive invariance, permanence, global attractivity, periodicity and almost periodic nature of system \eqref{eq:2} have been inspected. Results are as shown in Figures \ref{fig:1}-\ref{fig:6}. The same has been repeated for system \eqref{eq:k} and results are as shown in Figures \ref{fig:8} and \ref{fig:9}. 

\par From Figures \ref{fig:1}(a) and \ref{fig:3}(a) it is noticed that two trajectories starting at different initial values (for initial conditions and parameter values refer to Table \ref{Table:1}) tend to move in a bounded region and remain in it even as $t$ tends to infinity and thus using Theorems \ref{thm:t1} and \ref{thm:t2}, one can show positive invariance and permanence of system \eqref{eq:2}. Using theorem \ref{thm:t3}, in Figures \ref{fig:1}(b),(c) and \ref{fig:3}(b),(c) time series and phase portraits of globally attractive solutions of system \eqref{eq:2} are shown. The highlighted part of the trajectories in Figures \ref{fig:1}(c) and \ref{fig:3}(c) clearly shows the solution of system \eqref{eq:2} that is globally attractive. Figures \ref{fig:2}(a) and \ref{fig:4}(a),(b) depict periodic solutions to system \eqref{eq:2} and for parameters with different rationally independent periods, existence of almost periodic solutions is clearly seen in Figures \ref{fig:2}(b) and \ref{fig:4}(c),(d) (see Table \ref{Table:1} for parameter values and initial conditions) as proved in Theorems \ref{thm:t4} and \ref{thm:t5}. Using Theorems \ref{thm:t6},\ref{thm:t7},\ref{thm:t8},\ref{thm:t9} and \ref{thm:t10} the same experiment has been repeated for system \eqref{eq:k} as well. Results are shown in Figures \ref{fig:8} and \ref{fig:9} (for parameter values and initial conditions see Table \ref{Table:1}).

\par For various values of $\varsigma$ (delay in maturity), dynamics of system \eqref{eq:2} are observed in Figure \ref{fig:5} (other parameters are mentioned in Table \ref{Table:1}). In Figure \ref{fig:5}(a) when $\varsigma=0$ (i.e, hilsa spawns instantaneously without any delay) it can be seen that hilsa is abundant and thus the density of eel increases too. When $\varsigma=2$, $\varsigma =2.5$ and $\varsigma =3$ it can be seen that there is a decline in the prey density in Figures \ref{fig:5}(b),(c) and (d) which leads to a fall in predator population. Moreover, a trajectory to system \eqref{eq:2} starting at a particular initial value and having parameters with different rationally independent periods tends to move in a bounded region but it may follow a complicated path. This is depicted in Figures \ref{fig:6}(a) and (b) (for parameter values and initial conditions refer to Table \ref{Table:1}). Also, for various $\varsigma $, $\zeta$ and $A$, extreme changes in the behaviour of model \eqref{eq:2} are observed from Figure \ref{fig:7}. The same has been tabulated in Table \ref{Table:2}. Time series when $\zeta$ takes a very low value (i.e, eel depends only on alternative food) can be observed from Figure \ref{fig:10}. In this figure, it can be seen that as $A(t)$ varies between different time dependent/periodic functions, the eel population differs both for system \eqref{eq:2} and system \eqref{eq:k} in Figures \ref{fig:10}(a),(b),(c) and \ref{fig:10}(d),(e),(f) respectively. It is observed that if $A(t)$ is very less then the eel population decreases gradually where as if the amount of alternative food is high then eel survives.

\section{Conclusion}\label{s9}
The ecological relationship between hilsa and eel in the marine water and eel's reliance on alternative food during the monsoon when hilsa migrates to GBM river basin for spawning are modeled as a dynamical system by considering the temporal inhomogeneity of the parameters involved. Also, this model deals with the age based growth of hilsa along with the predation term of Beddington-DeAngelis type in time variant parameters.
\par Some sufficient conditions for positive invariance and permanence are obtained from Theorems \ref{thm:t1} and \ref{thm:t2} respectively. It is observed that if $\varepsilon=0$ in Theorem \ref{thm:t1}, it results in the condition in Theorem \ref{thm:t2}. This indicates that if $\Gamma_\varepsilon$ is positively invariant in system \eqref{eq:2}, then the system \eqref{eq:2} must be permanent, which is clearly seen in Figures \ref{fig:1} and \ref{fig:3}. It means that if the populations of hilsa and eel lie within a bounded set and remain forever in that particular set then their populations would never become extinct. In Theorem \ref{thm:t3}, relevant criteria for global stability of a bounded positive solution are established by formulating a Lyapunov function.
\par In addition, we obtained criterion for the existence of a positive $\rho$-periodic solution in Theorem \ref{thm:t4} using continuation theorem. Theorem \ref{thm:t4} provides the range of existence for the periodic solution. Also, system \eqref{eq:2} possesses a distinctive almost periodic solution as seen in Figures \ref{fig:2} (b) and \ref{fig:4} (c), (d) and the same is derived in Theorem \ref{thm:t5}. Moreover, some complicated trajectories of system \eqref{eq:2} are observed in Figures \ref{fig:6} (a) and (b), whose behaviour is unpredictable. Further, table \ref{Table:2} clearly shows that the predator population shall face the risk of extinction due to lack of alternative food during the monsoon season.
\par From Figure \ref{fig:5} (a) it is evident that if there is no delay in maturity ($\varsigma = 0$), i.e., hilsa spawns immediately after its birth, then the population of hilsa grows immensely. Thus, eel gets excess food to feed on. This results in a rapid increase in eel population as well. Such a scenario is not ecologically beneficial. On the other hand if the maturity delay in hilsa is neglected, i.e, immature hilsa are continued to be harvested, then hilsa shall face the risk of extinction. In any food chain if one species goes extinct there will be an adverse effect on other species. That is, the abundance of eel may decline. Whereas, when $\varsigma$ is considered to be $2, 2.5$ and $3$ the prey-predator population is positively invariant, permanent, globally stable and periodic as seen in Figures \ref{fig:5} (b), (c) and (d), thereby suggesting that incorporating age based growth model of hilsa can be of great benefit to the ecosystem.
\par Moreover, system \eqref{eq:2}, which handles the age-structured growth in prey population can be extended to system \eqref{eq:k} that can handle both age-structured growth and age-selective harvesting of prey. Various dynamical aspects of system \eqref{eq:k} have been examined in Theorems \ref{thm:t6}-\ref{thm:t10} and verified using Figures \ref{fig:8}-\ref{fig:9}. Samanta \cite{SGP10} mentioned about the ineffectiveness of the time delay on permanence. But time delay does effect the global attractivity of a non-autonomous dynamical system. System \eqref{eq:2} only deals with $\varsigma_1$ whereas system \eqref{eq:k} deals with $\varsigma_1$ and $\varsigma_2$. But this additional delay $\varsigma_2$ in system \eqref{eq:k} has no effect on permanence (when compared to system \eqref{eq:2}) as observed in Figures \ref{fig:3}(a) and \ref{fig:8}(a). Inclusion of delay in harvest ($\varsigma_2$) in system \eqref{eq:k} does affect the global attractivity (when compared to system \eqref{eq:2}). The same can be observed in Figures \ref{fig:3}(b),(c) and \ref{fig:8}(b),(c). Thus, for greater benefit of the ecosystem both $\varsigma_1$ (delay in maturity) and $\varsigma_2$ (delay in harvest) can be treated with equal importance.

\begin{figure}[]\centering
\subfloat[]{\includegraphics[width=6.7in,height=3.5in]{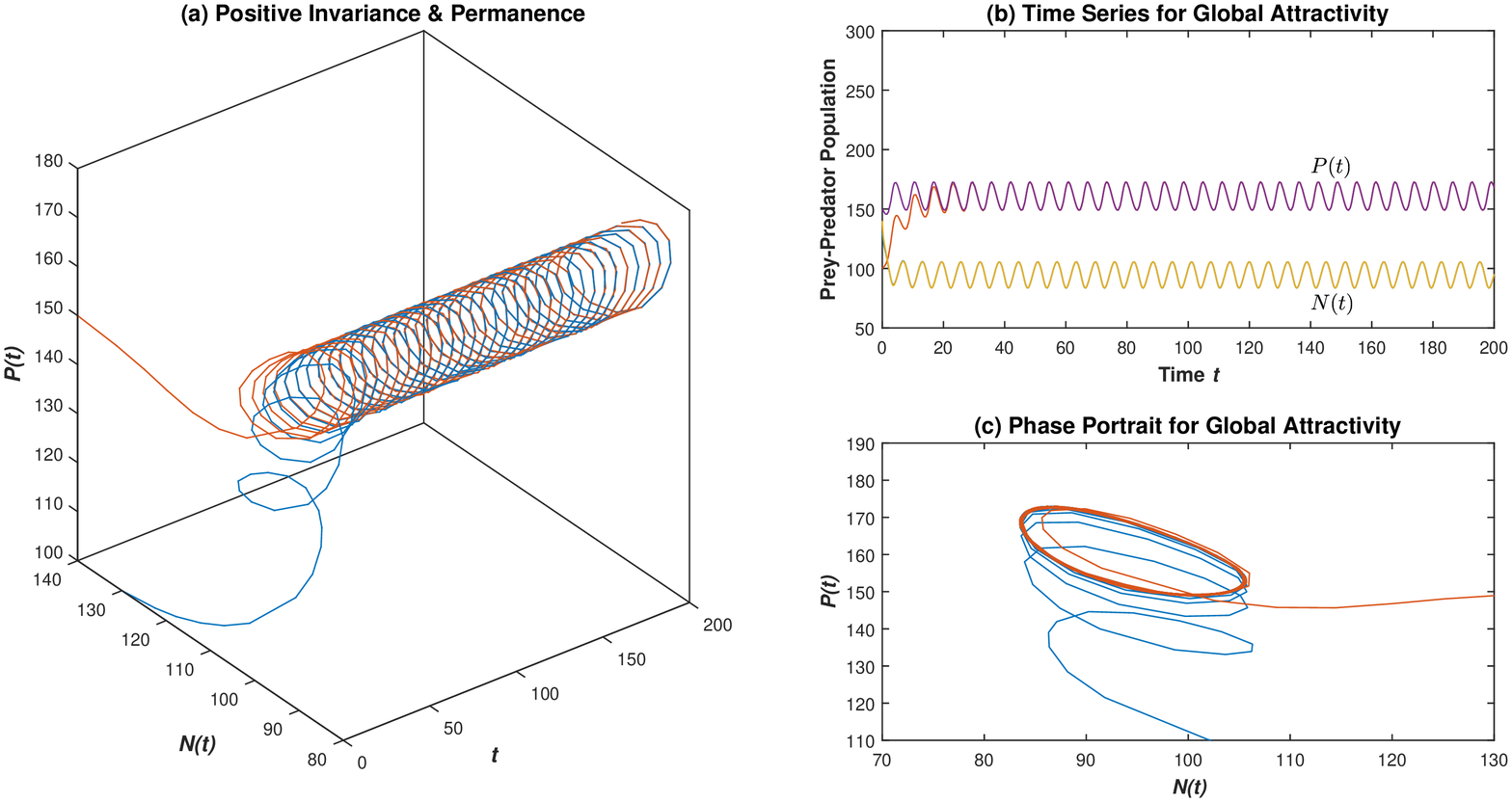}}\vspace{-0.8cm}\caption{(a) Phase space diagram of system \eqref{eq:2} showing positive invariance and permanence. (b) and (c) are time series plot and phase portrait, respectively, for a globally attractive solution of system \eqref{eq:2}.} \label{fig:1}
\end{figure}

\begin{figure}\centering
\subfloat[]{\includegraphics[width=6.7in,height=3.5in]{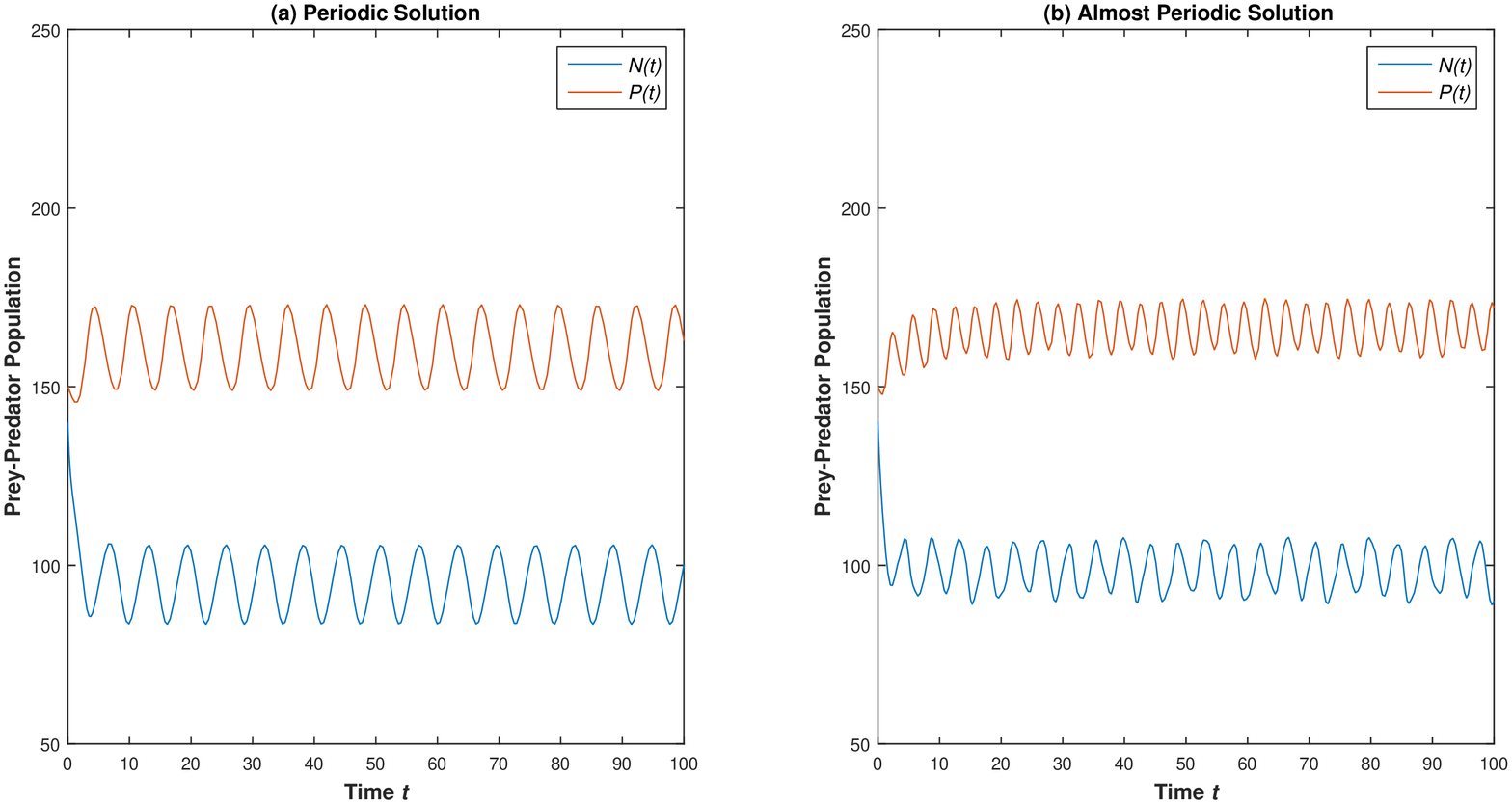}}\vspace{-0.8cm}\caption{Time series plots for Hilsa-Eel population of system \eqref{eq:2}. Figure (a) and (b) show a periodic and almost periodic solution to system \eqref{eq:2}, respectively.} \label{fig:2}
\end{figure}

\begin{figure}[]\centering
\subfloat[]{\includegraphics[width=6.7in,height=3.5in]{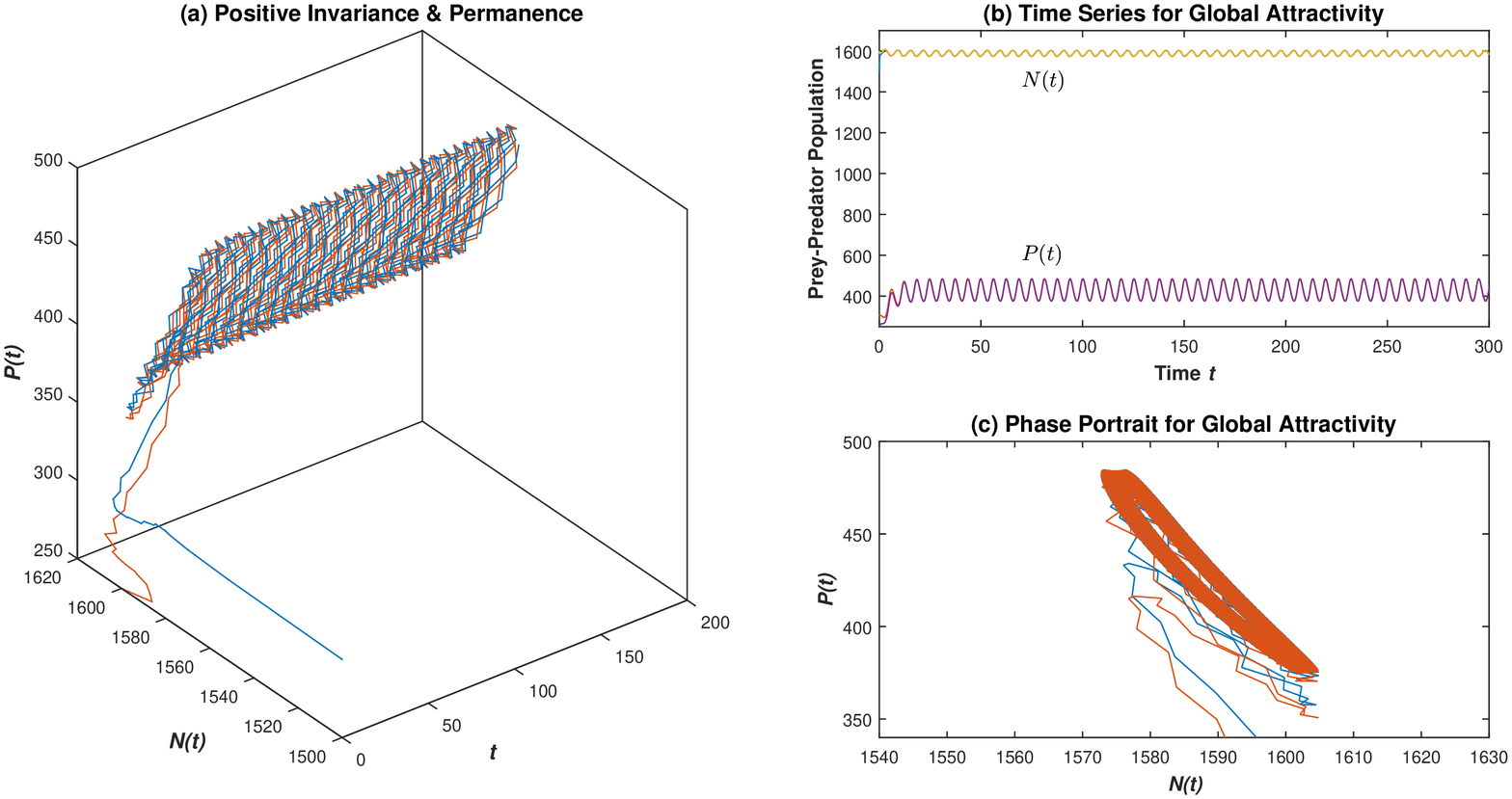}}\vspace{-0.8cm}\caption{(a) Phase space diagram of system \eqref{eq:2} showing positive invariance and permanence. (b) and (c) are time series plot and phase portrait, respectively, for a globally attractive solution of system \eqref{eq:2}.} \label{fig:3}
\end{figure}

\begin{figure}\centering
\subfloat[]{\includegraphics[width=6.7in,height=3.5in]{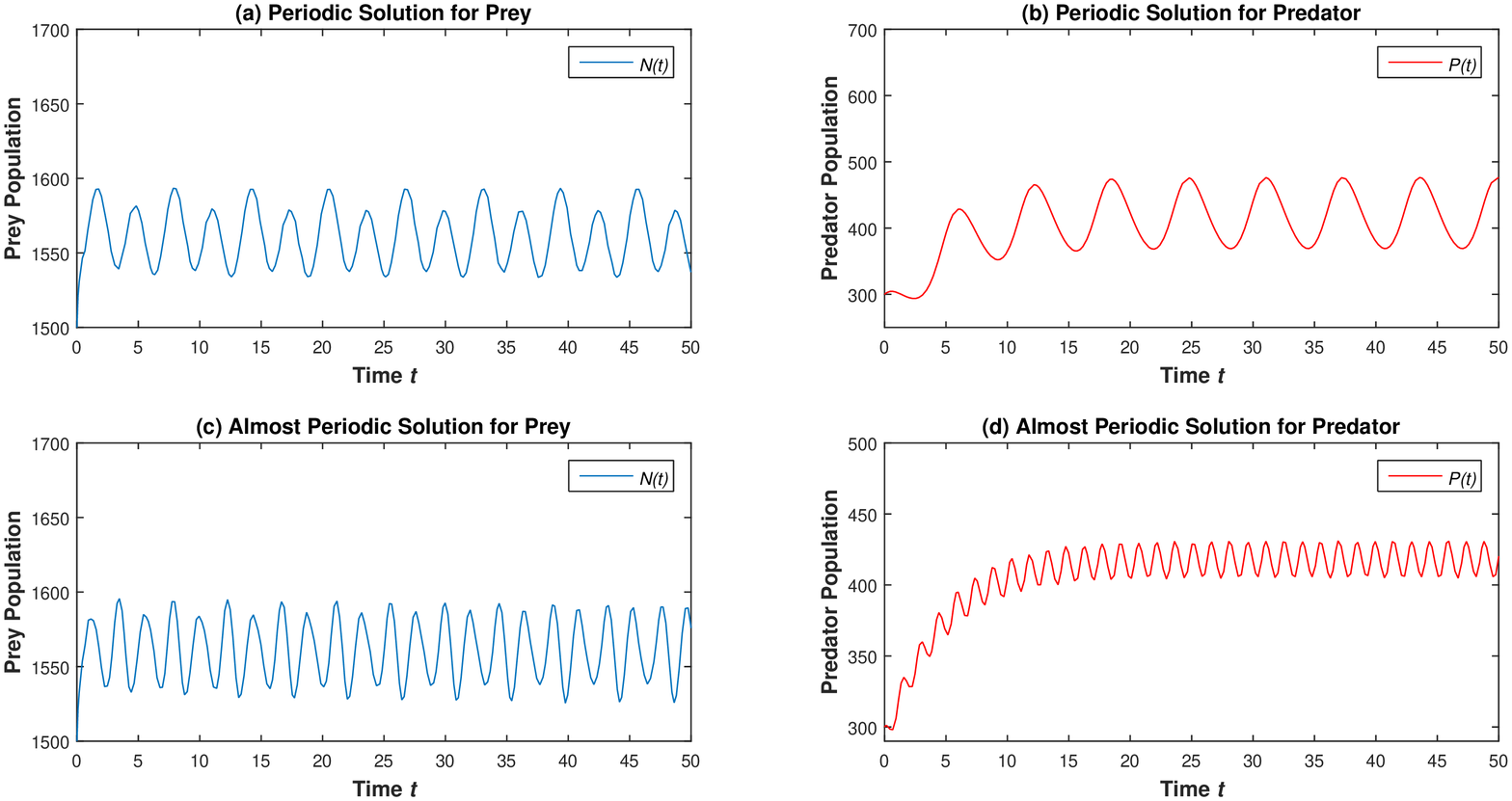}}\vspace{-0.8cm}\caption{Time series plots for Hilsa-Eel population of system \eqref{eq:2}. Figure (a),(b) and (c),(d) show periodic and almost periodic solutions to system \eqref{eq:2} respectively.} \label{fig:4}
\end{figure}

\begin{figure}\centering
\subfloat[]{\includegraphics[width=6.7in,height=3.5in]{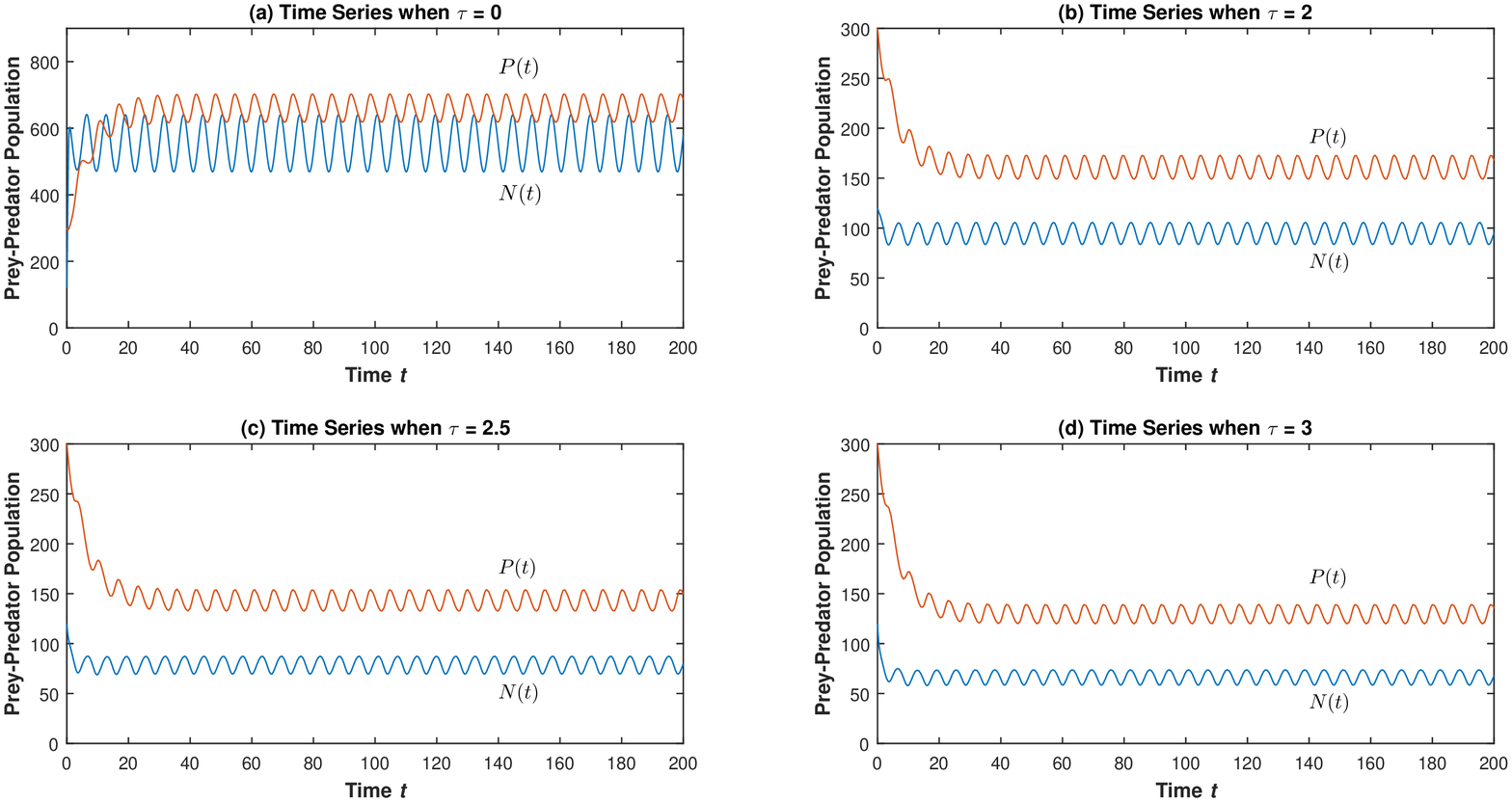}}\vspace{-0.8cm}\caption{Time series diagram of system \eqref{eq:2} when $\varsigma=0,2,2.5$ and $3$. }\label{fig:5}\vspace{0.2cm}
\end{figure}

\begin{figure}\centering
\subfloat[]{\includegraphics[width=6.7in,height=3.5in]{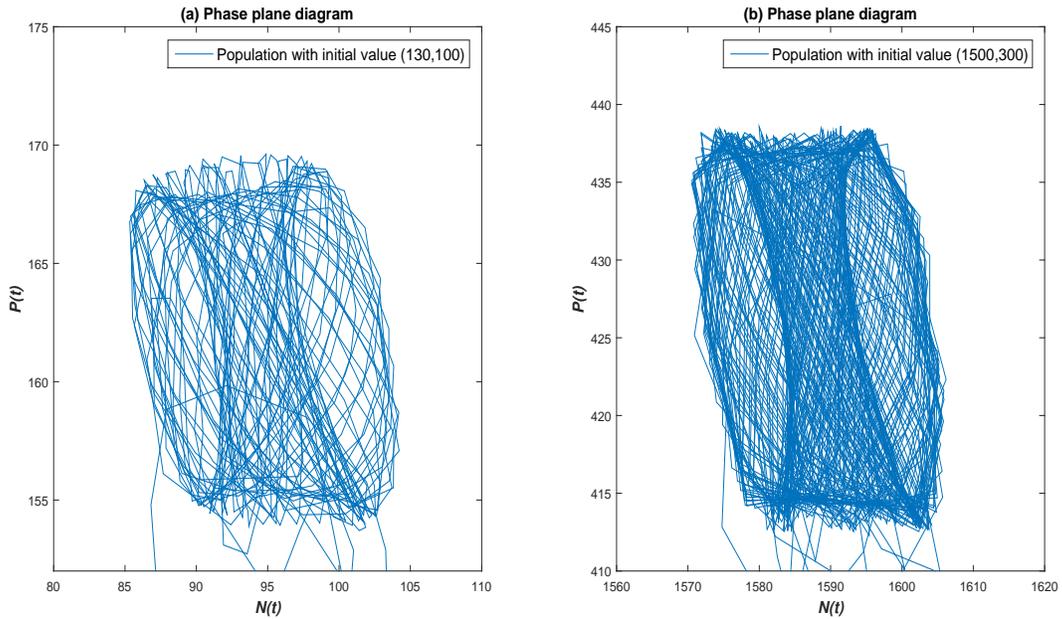}}\vspace{-0.8cm}\caption{Phase portrait of a complicated trajectory of system \eqref{eq:2} with almost periodic coefficients.} \label{fig:6}\vspace{0.2cm}
\end{figure}

\begin{figure}\centering
\subfloat[]{\includegraphics[width=6.7in,height=3.5in]{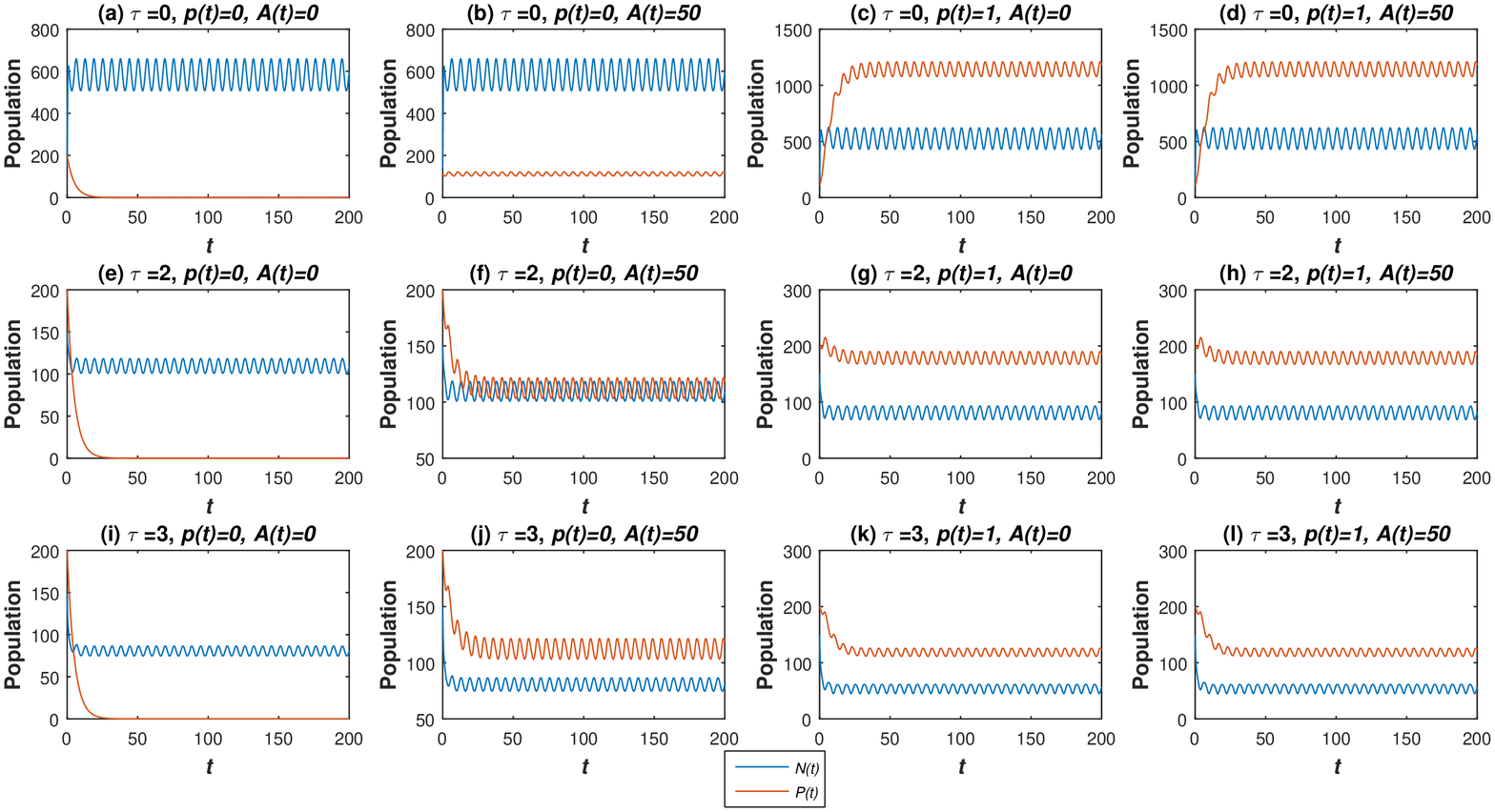}}\vspace{-0.4cm}\caption{Time series plots of trajectories of system \eqref{eq:2} with different values of $\varsigma, \zeta(t)$ and $A(t)$.}\label{fig:7}\vspace{0.2cm}
\end{figure}

\begin{figure}[]\centering 
\vspace{-0.2cm}\subfloat[]{\includegraphics[width=6.7in,height=3.5in]{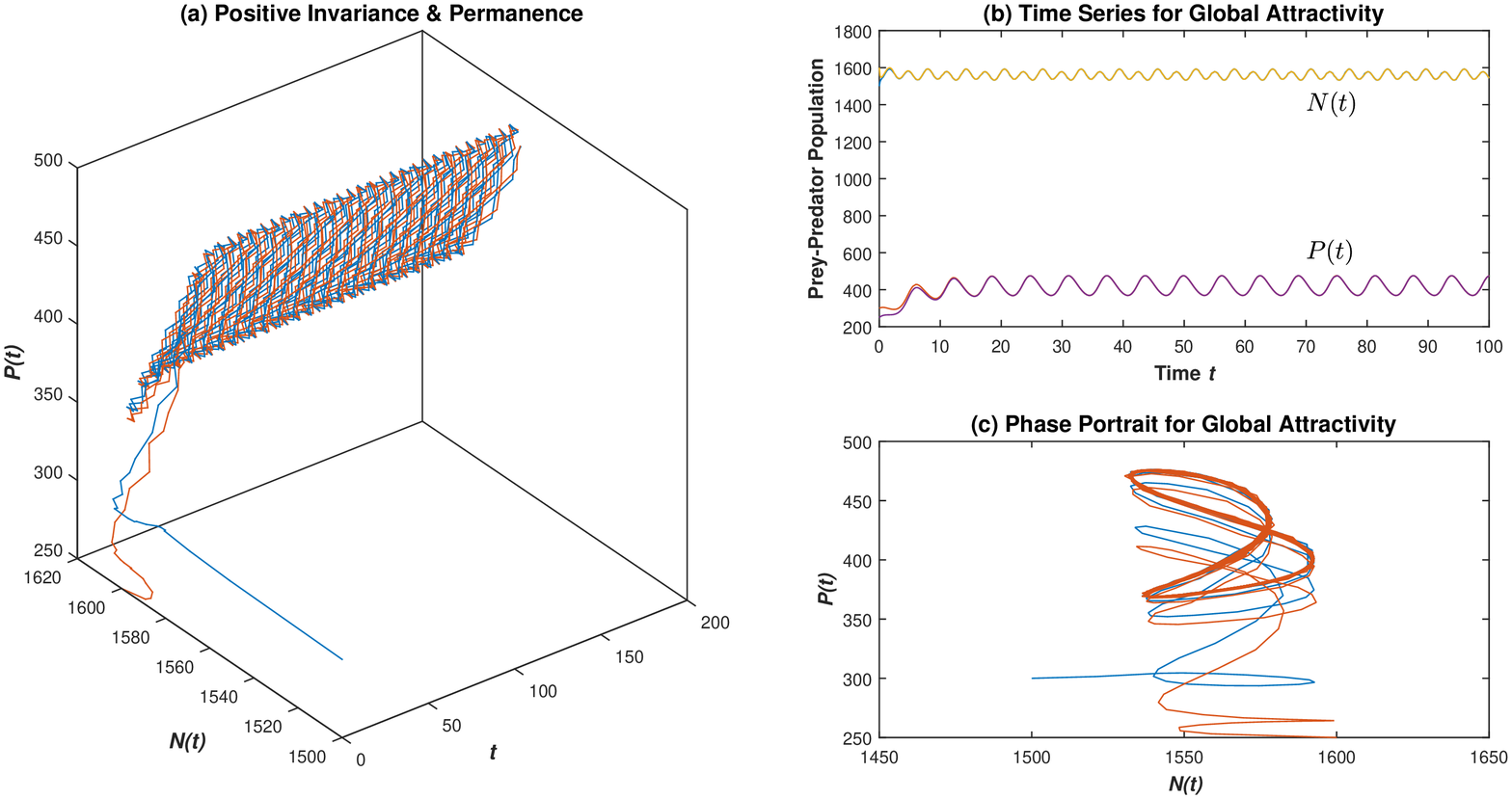}}\vspace{-0.8cm}\caption{Population curves of system \eqref{eq:k} with harvesting delay. (a) Phase space diagram of system \eqref{eq:k} showing positive invariance and permanence. (b) and (c) are time series plot and phase portrait, respectively, for a globally attractive solution of system \eqref{eq:k}.}\label{fig:8}
\end{figure}

\begin{figure}\centering
\subfloat[]{\includegraphics[width=6.7in,height=3.5in]{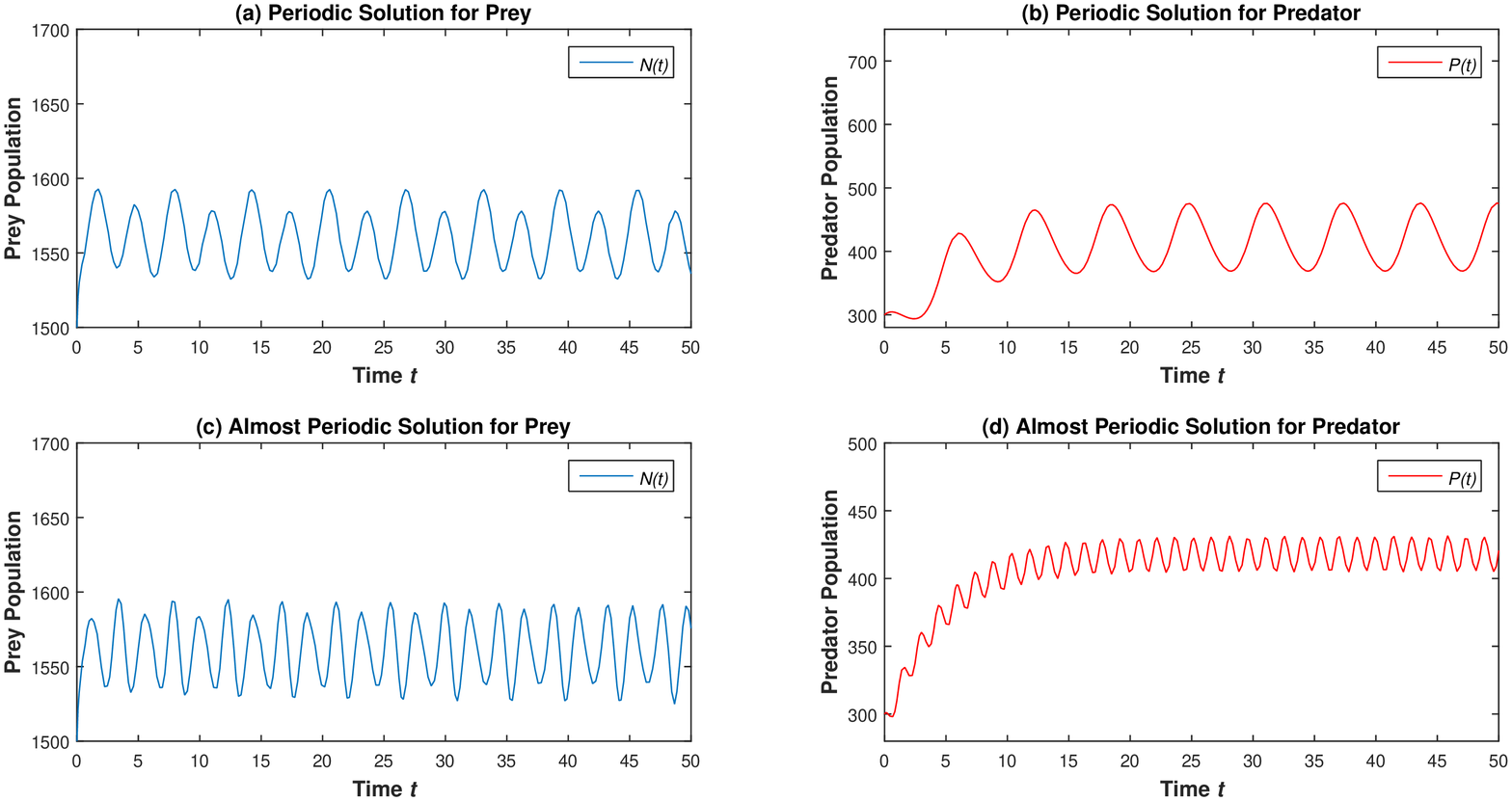}}\vspace{-0.8cm}\caption{Time series plots for system \eqref{eq:k}. Figure (a),(b) and (c),(d) show the existence of periodic and almost periodic solutions to system \eqref{eq:k} respectively.}\label{fig:9}\vspace{0.2cm}
\end{figure}

\begin{figure}\centering
\subfloat[]{\includegraphics[width=6.7in,height=3.5in]{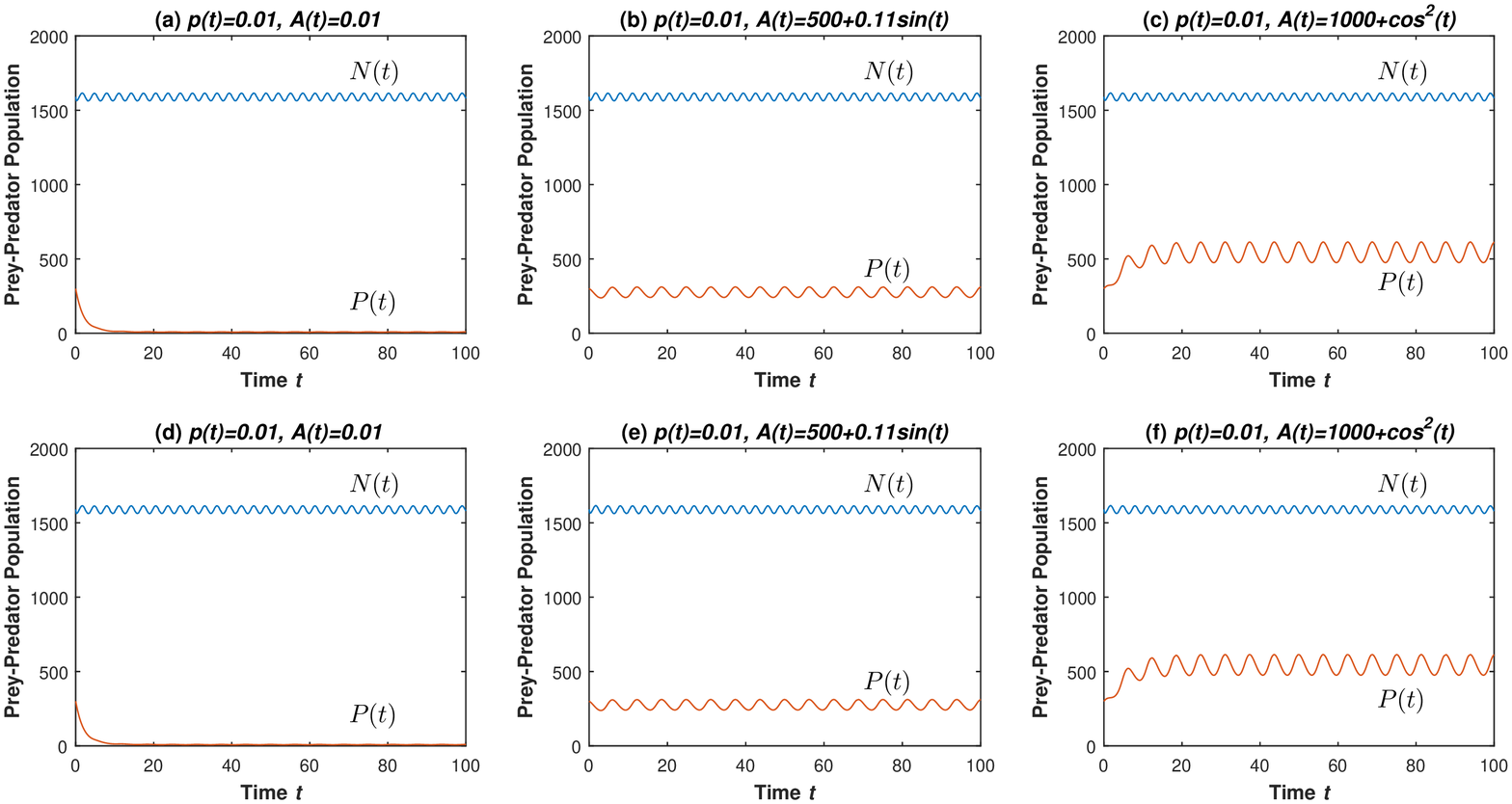}}\vspace{-0.8cm}\caption{(a),(b),(c) and (d),(e),(f) depict the time series plots for various values of $A(t)$ in systems \eqref{eq:2} and \eqref{eq:k} respectively.} \label{fig:10}\vspace{0.2cm}
\end{figure}

\begin{table}[]
\centering
\caption{Parameter sets for Figures \ref{fig:1}-\ref{fig:10}}
\label{Table:1}\scalebox{0.96}{
\begin{tabular}{|p{2cm}|p{14.5cm}|}\hline
\textbf{Figure} & \textbf{Parameters and Initial conditions}\\ \hline
Fig. \ref{fig:1} & $b(t)=6+0.125(cos(t))$, $c(t)=0.01-0.0011(cos(t))$, $\varphi(t)=0.25$, $\zeta(t)=0.5$, $\nu(t)=1.2+0.10(sin(t))$, $\kappa(t)=1$, $\xi(t)=2+cos(t)$, $m(t)=0.20$, $A(t)=50$, $a(t)=0.025$ and $\varsigma=2$. Initial conditions: $\chi(0)=130,
\Upsilon(0)=100$ and $\chi(0)=140, \Upsilon(0)=150$.\\ \hline
Fig. \ref{fig:2} & $b(t)=6+0.125(cos(t))$, $c(t)=0.01-0.0011(cos(t))$, $\varphi(t)=0.25$, $\zeta(t)=0.5$, $\nu(t)=1.2+0.10(sin(t))$, $\kappa(t)=1$, $\xi(t)=2+cos(t)$, $m(t)=0.20$, $A(t)=50$, $a(t)=0.025$, $\varsigma=2$ and $\rho=2 \pi$. Initial conditions:
$\chi(0)=130, \Upsilon(0)=100$ and $\chi(0)=140, \Upsilon(0)=150$. In Fig. \ref{fig:2} (b), only $b(t)=6+0.125(cos(\sqrt2 \pi t))$, $c(t)=0.01-0.0011(cos(\sqrt2 t))$, 
$\nu(t)=1.2+0.1(sin(7 \sqrt2 t))$ and $\xi(t)=2+cos(\sqrt7 t)$. 
\\ \hline
Fig. \ref{fig:3} & $b(t)=170+sin(t)$, $c(t)=0.0029$, $\varphi(t)=0.02+0.01(cos(t))$, $\zeta(t)=0.5$, $\nu(t)=1.2$, $\kappa(t)=1$, $\xi(t)=2+si\chi(t)$, $m(t)=0.5$, $A(t)=0.01$, $a(t)=0.01$ and $\varsigma=2$. Initial conditions: $\chi(0)=1500$, $\Upsilon(0)=300$ and
$\chi(0)=1600$, $\Upsilon(0)=250$. \\ \hline
Fig. \ref{fig:4} & $b(t)=170+sin(t)$, $c(t)=0.0029$, $\varphi(t)=0.02+0.01(cos(t))$, $\zeta(t)=0.5$, $\nu(t)=1.2$, $\kappa(t)=1$, $\xi(t)=2+si\chi(t)$, $m(t)=0.5$, $A(t)=0.01$, $a(t)=0.01$, $\varsigma=2$ and $\rho=2 \pi$. Initial condition: $\chi(0)=1500$, $\Upsilon(0)=300$.
In Fig. \ref{fig:4} (c),(d), only $\varphi(t)=0.02+0.01(cos(3 \sqrt7 t))$, $\xi(t)=2+sin((2.5) \sqrt3 t)$. 
\\ \hline
Fig. \ref{fig:5} & $b(t)=6+0.125(cos(t))$, $c(t)=0.01-0.0011(cos(t))$, $\varphi(t)=0.25$, $\zeta(t)=0.5$, $\nu(t)=1.2+0.10(sin(t))$, $\kappa(t)=1$, $\xi(t)=2+\cos(t)$, $m(t)=0.20$, $A(t)=50$, $a(t)=0.025$ and $\varsigma= 0, 2, 2.5$ and $3$. Initial condition: $(120,
300)$. \\ \hline
Fig. \ref{fig:6} & In (a) $b(t)=6+0.125(cos(\sqrt2 \pi t))$, $c(t)=0.01-0.0011(cos(\sqrt2 t))$, $\varphi(t)=0.25$, $\zeta(t)=0.5$, $\nu(t)=1.2+0.1(sin(7 \sqrt2 t))$, $\kappa(t)=1$, $\xi(t)=2+cos(\sqrt7 t)$, $m(t)=0.20$, $A(t)=50$, $a(t)=0.025$, $\varsigma=2$ and $
\rho=2 \pi$. In (b) $b(t)=170+sin(t)$, $c(t)=0.0029$, $\varphi(t)=0.02+0.01(cos(3 \sqrt7 t))$, $\zeta(t)=0.5$, $\nu(t)=1.2$, $\kappa(t)=1$, $\xi(t)=2+sin((2.5) \sqrt3 t)$, $m(t)=0.5$, $A(t)=0.01$, $a(t)=0.01$, $\varsigma=2$ and $\rho=2 \pi$. \\ \hline
Fig. \ref{fig:7} & $b(t)=6+0.125(cos(t))$, $c(t)=0.01-0.0011(cos(t))$, $\varphi(t)=0.25$, $\nu(t)=1.2+0.10(sin(t))$, $\kappa(t)=1$, $\xi(t)=2+\cos(t)$, $m(t)=0.20$, $a(t)=0.025$ and various values of $\varsigma, \zeta(t)$ and $A(t)$. Initial condition: (150,200). \\ 
\hline
Fig. \ref{fig:8} & $b(t)=170+sin(t)$, $c(t)=0.0029$, $\varphi(t)=0.02+0.01(cos(t))$, $\zeta(t)=0.5$, $\nu(t)=1.2$, $\kappa(t)=1$, $\xi(t)=2+sin(t)$, $m(t)=0.5$, $A(t)=0.01$, $a(t)=0.01$, $q(t)=0.01+cos(t)/300$, $E(t)=2+sin(t)$, $\varsigma_1=2$ and $
\varsigma_2=1$. Initial conditions: $\chi(0)=1500$, $\Upsilon(0)=300$ and $\chi(0)=1600$, $\Upsilon(0)=250$. \\ \hline
Fig. \ref{fig:9} & $b(t)=170+sin(t)$, $c(t)=0.0029$, $\varphi(t)=0.02+0.01(cos(t))$, $\zeta(t)=0.5$, $\nu(t)=1.2$, $\kappa(t)=1$, $\xi(t)=2+sin(t)$, $m(t)=0.5$, $A(t)=0.01$, $a(t)=0.01$, $q(t)=0.01+cos(t)/300$, $E(t)=2+sin(t)$, $\varsigma_1=2$ and $\varsigma_2=1$. 
Initial condition: $\chi(0)=1500$, $\Upsilon(0)=300$.
In Fig. \ref{fig:9} (c),(d), only $\varphi(t)=0.02+0.01(cos(3 \sqrt7 t))$, $\xi(t)=2+sin((2.5) \sqrt3 t)$. 
\\ \hline
Fig. \ref{fig:10} & In (a,b,c): $b(t)=170+sin(t)$, $c(t)=0.0029$, $\varphi(t)=0.02+0.01(cos(t))$, $\zeta(t)=0.01$, $\nu(t)=1.2$, $\kappa(t)=1$, $\xi(t)=2+sin(t)$, $m(t)=0.5$, $a(t)=0.01$, $\varsigma=2$ and various values of $A(t)$. Initial condition: $\chi(0)=1500$, 
$\Upsilon(0)=300$. In (d,e,f) $b(t)=170+sin(t)$, $c(t)=0.0029$, $\varphi(t)=0.02+0.01(cos(t))$, $\zeta(t)=0.01$, $\nu(t)=1.2$, $\kappa(t)=1$, $\xi(t)=2+sin(t)$, $m(t)=0.5$, $a(t)=0.01$, $q(t)=0.01+cos(t)/300$, $E(t)=2+si\chi(t)$, $\varsigma_1=2$, $\varsigma_2=1$ and various 
values of $A(t)$. Initial condition: $\chi(0)=1500$, $\Upsilon(0)=300$. \\ \hline
\end{tabular}}
\end{table}

\begin{table}[]
\centering
\caption{Prey-predator population for various values of $\varsigma$, $\zeta(t)$ and $A(t)$ in system \eqref{eq:2}}
\label{Table:2}
\begin{tabular}{|p{2cm}|p{0.8cm}|p{0.8cm}|p{10cm}|p{1.2cm}|}\hline
$\mathbf{\varsigma}$ & $\mathbf{\zeta(t)}$ & $\mathbf{A(t)}$ & \textbf{Nature of the system} & \textbf{Figure}\\ \hline
$\varsigma=0$ & $0$ & $0$ & \text{Eel depends only on alternative food which} & \text{\ref{fig:7} (a)}\\
\text{(Hilsa } & & &\text{does not exist. Thus eel becomes extinct.}& \\ \cline{2-5}
\text{spawns}& $0$ & $50$ & \text{Eel depends solely on the limited alternative food} & \text{\ref{fig:7} (b)}\\
\text{instantly} & & &\text{and thus manages to sustain.} &\\ \cline{2-5}
\text{after birth.}& $1$ & $0$ & \text{Eel feeds only on hilsa. Thus the lack of alternative} & \text{\ref{fig:7} (c)}\\
\text{i.e, Prey} & & &\text{food does not effect it. Overabundance of prey } &\\
\text{population} & & &\text{leads to superabundance of the predator.} &\\
\cline{2-5}
\text{grows} & $1$ & $50$ & \text{Eel feeds only on hilsa. Limited amount of alternative} & \text{\ref{fig:7} (d)}\\ 
\text{immensely.)} & & &\text{food does not affect the predator. Overabundance} &\\ 
& & &\text{of prey leads to superabundance of the predator.} &\\\hline 
$\varsigma=2$ & $0$ & $0$ & \text{The sole dependence of eel on alternative food which} & \text{\ref{fig:7} (e)}\\ 
\text{(Maturity}& & &\text{does not exist, leads to the extinction of the predator.} &\\ \cline{2-5}
\text{delay in}& $0$ & $50$ & \text{Eel depends only on alternative food (which is} & \text{\ref{fig:7} (f)}\\ 
\text{prey }& & &\text{limited). Thus, eel manages to sustain.} &\\ \cline{2-5}
\text{population.)}& $1$ & $0$ & \text{Eel feeds only on hilsa and thus the lack of alternative} & \text{\ref{fig:7} (g)}\\
& & &\text{food does not affect the predator. Prey-predator} &\\
& & & \text{population is well balanced.} & \\ \cline{2-5}
& $1$ & $50$ & \text{Eel feeds only on hilsa. Thus, the limited availability} & \text{\ref{fig:7} (h)}\\
& & & \text{of alternative food does not affect the predator.}& \\
& & & \text{Prey-predator population is well balanced.}&\\ \hline 
$\varsigma=3$ & $0$ & $0$ & \text{The sole dependence of eel on alternative food (which} & \text{\ref{fig:7} (i)}\\
\text{(Delay in} & & & \text{does not exist) leads to the extinction of the predator.} & \\ \cline{2-5}
\text{maturity is} & $0$ & $50$ & \text{The sole dependence of predator on limited alternative} & \text{\ref{fig:7} (j)}\\ 
\text{a little} & & & \text{food allows predator sustain.} & \\
\cline{2-5}
\text{longer. A} & $1$ & $0$ & \text{Eel feeds only on hilsa. Thus the limited availability} & \text{\ref{fig:7} (k)}\\ 
\text{slight} & & & \text{of alternative food does not affect the predator.} & \\
\text{decline is} & & & \text{Prey-predator population is well balanced. A slight} &\\
\text{observed} & & & \text{decline is observed in eel population when compared to} & \\
\text{in prey} & & & \text{Figure \ref{fig:7} (g).} &\\
\cline{2-5}
\text{population} & $1$ & $50$ & \text{Eel consumes only hilsa and thus the limited} & \text{\ref{fig:7} (l)}\\
\text{when} & & & \text{availability of alternative food does not have any} & \\
\text{compared} & & & \text{impact on the predator. Prey-predator population is} &\\
\text{to Figures \ref{fig:7}} & & & \text{well balanced but, a slight decline is observed} & \\
\text{(e, f, g, h).)} & & & \text{in eel population when compared to Figure \ref{fig:7} (h).} & \\
\hline 
\end{tabular}
\end{table}
\end{document}